\documentclass[11pt, 3p]{elsarticle}
\makeatletter
\def\ps@pprintTitle{%
 \let\@oddhead\@empty
 \let\@evenhead\@empty
 \def\@oddfoot{\centerline{\thepage}}%
 \let\@evenfoot\@oddfoot}
\makeatother

\date{}
\def\texpsfig#1#2#3{\vbox{\kern #3\hbox{\includegraphics{#1}\kern #2}}\typeout{(#1)}}

\usepackage{url}
\usepackage[hidelinks]{hyperref}
\usepackage{bbm}
\usepackage{eurosym}                           
\usepackage[latin1]{inputenc}                  
\usepackage{url}                               
\usepackage{longtable}                         
\usepackage{array}                             

\usepackage{amsthm}
\usepackage{amsbsy}

\usepackage{placeins}  

\usepackage{amssymb}
\usepackage{amsmath}
\usepackage{enumerate}

\usepackage{graphicx,color}
\usepackage{epsfig}
\usepackage[ruled,vlined]{algorithm2e}
\usepackage{multirow,bigdelim}
\usepackage[table]{xcolor}
\usepackage{natbib}
\usepackage{cleveref}

\theoremstyle{plain}
\newtheorem{thm}{Theorem}[section]

\newtheorem{dfn}[thm]{dfn}

\newtheorem{rem}{rem}[section]

\theoremstyle{rem}

\theoremstyle{plain}
\newtheorem{lem}[thm]{lem}

\theoremstyle{dfn}

\newcommand{\e}{{\rm e}}        
\def\R{\mathbb{ R}}             
\def\E{\mathbb{ E}}             
\def\Q{\mathbb{ Q}}             

\def\P{\mathbb{ P}}             

\def\Cov{\mathrm{\mathbb{C}ov}}
\def\corr{\mathrm{corr}}            

\def\Var{\mathrm{\mathbb{V}ar}}   
\renewcommand{\d}{{\rm d}}      
\def\dt{{\rm d}t}
\def\ds{{\rm d}s}

\def\e{{\mathrm{e}}}
\def\1{{\mathbbm{1}}}            

\def\PnL{\mathrm{PnL}}         
\def\CTD{\mathrm{CTD}}         
\def\CF{\mathrm{cf}}        
\def\det{\mathrm{det}}     
\def\stoch{\mathrm{stoch}} 
\def\none{\mathrm{none}}

\theoremstyle{plain}

\usepackage[margin=1cm]{caption}

\numberwithin{equation}{section}	     

\title{Sensitivities and Hedging of the Collateral Choice Option}

\begin{document}

\author[1]{Griselda Deelstra}
\ead{Griselda.Deelstra@ulb.be}
\author[2,3]{Lech A.\ Grzelak}
\ead{L.A.Grzelak@uu.nl}
\author[1]{Felix L.\ Wolf}
\ead{Felix.Wolf@ulb.be}

\address[1]{Department of Mathematics, Universit\'e libre de Bruxelles, Brussels, Belgium}
\address[2]{Mathematical Institute, Utrecht University, Utrecht, the Netherlands}
\address[3]{Rabobank, Utrecht, the Netherlands}

\begin{abstract}\small{
The collateral choice option allows a collateral-posting party the opportunity to change the type of security in which the collateral is deposited. Due to non-zero collateral basis spreads, this optionality significantly impacts asset valuation.
Because of the complexity of valuing the option, many practitioners resort to deterministic assumptions on the collateral rates. 
In this article, we focus on a valuation model of the collateral choice option based on stochastic dynamics.
Intrinsic differences in the resulting collateral choice option valuation and its implications for collateral management are presented. 
We obtain sensitivities of the collateral choice option price under both the deterministic and the stochastic model, and we show that the stochastic model attributes risks to all involved collateral currencies. Besides an inability to capture volatility effects, the deterministic model exhibits a digital structure in which only the cheapest-to-deliver currency influences the valuation at a given time. 
We further consider hedging an asset with the collateral choice option by a portfolio of domestic and foreign zero-coupon bonds that do not carry the collateral choice option. We propose static hedging strategies based on the crossing times of the deterministic model and based on variance-minimization under the stochastic model. We show how the weights of this model can be explicitly determined with the semi-analytical common factor approach and we show in numerical experiments that this strategy offers good hedging performance under minimized variance.
}
\end{abstract}


\maketitle

{\let\thefootnote\relax\footnotetext{The views expressed in this paper are the personal views of the authors and do not necessarily reflect the views or policies of their current or past employers.}}

\section{Introduction}
The collateral choice option describes an optionality in collateralized assets, where a basket of admissible collateral securities is available to the collateral posting party. 
Collateral posters are incentivized to optimize their choice of collateral security, as posting collateral is associated with procurement costs and returns generated by the collateral are transferred to the posting party. 
\par
{This article outlines} existing stochastic and deterministic models used in {evaluating} the collateral choice option. We show that the prices obtained from the models fundamentally differ and that the models do not exhibit the same type of sensitivities in their option {prices}. We further consider the hedging problem of an asset {equipped} with the collateral choice option when the hedging instruments may not carry the option themselves. We propose hedging strategies based on the stochastic and deterministic collateral choice models.
\par
We focus on the popular case of cash collateral with a choice between multiple available collateral currencies, but conclusions can be drawn for alternative securities like bonds or other assets.
In the choice between multiple collateral currencies, the \emph{collateral rates} of each currency are contractually determined {with} distinct interest rates. Based on the returns of these rates, a hierarchy of collateral currencies is implied{,} and an optimal choice of collateral currency, known as the \emph{cheapest-to-deliver} collateral currency, arises. Stochastic models to describe this choice and {its} valuation have been previously treated in \cite{PiterbargAntonov,PiterbargFunding,BieleckiRutkowski,FujiiRisk,Macey2011,McCloud,PiterbargCooking,SankovichZhu}. Within this article, we use the stochastic common factor collateral choice model proposed {by} \cite{WolfCF}.
\par
The impact of the collateral choice option is closely linked to the asset on which it acts. In full generality, the value of the collateral choice option is determined by interactions between the collateral currencies and the collateralized assets, {and} also by the collateralization modalities themselves. 
{For example,} collateral posting modalities may restrict the collateral poster to only change the collateral currency whenever the collateral account passes through zero (i.e.\ when the collateral account vanishes), or that previously posted collateral may not be exchanged, and only additionally posed collateral can be of different currency. The resulting fragmentation of the collateral account is known as \emph{sticky collateral} and {is} treated in \cite{PiterbargSticky}.
\par
The above restrictions result in an asset-specific collateral choice option, meaning the {option's price must} be individually calculated for each asset. In light of {a} large {number} of eligible assets, this is not computationally feasible for many Financial actors and motivates a model where the collateral choice option is valued separately from the affected assets. 
A common assumption in the literature which achieves this is called the \emph{free (collateral) substitution}. In this case, the entire collateral can be exchanged at any time. This makes it possible to persistently post the collateral in the currency with the highest (FX-adjusted) collateral rate at any time. Consequently, the optimal collateral rate is given by the maximum of all available rates.
\par
In this article, we further assume \emph{perfect collateralization}, meaning that the collateral account is updated in continuous time without frictions. This is considered a reasonable assumption under daily margin calls, {see e.g.\ }\cite{FujiiRisk}.
\par
{Moreover, a} perfectly collateralized asset is free from default risks as outstanding obligations are fully covered by the collateral account at any time, which simplifies discounting back to a single curve framework where only the funding value has to be considered. For this setting, it is shown in \cite{PiterbargCooking} that an adaption of the standard risk-free model dictates the collateral rate to be used for discounting.
\par
Thus, in the simplest case of one available collateral currency with the collateral rate $r_0$, a perfectly collateralized asset $U$ with a singular cash flow at time $T$ has, at time $t_0 \leq t \leq T$, the value
\begin{equation}\label{eq:singlecoll}
U(t) = \E_t^{\Q_0}\bigl[\e^{-\int_t^T r_0(s)\d s} U(T)\bigr],
\end{equation}
where $\Q_0$ is the appropriate measure {linked} with num\'eraire $B(t) = \exp(\int_{t_0}^t r_0(s)\d s)$ under which the collateralized asset is priced.\footnote{Throughout this article, we denote by $\E_t$ the conditional expectation with respect to the filtration at time $t$. The expectation with respect to initial time $t_0$ is indicated by $\E$.}
\par
We {equip} the above asset with the collateral choice option and free substitution rights between $N+1$ collateral currencies, where collateral posted in currency $i\in\{0, \dots, N\}$ accrues interest at rate $r_i$. We assume that the foreign interest rates $r_i$, $i\geq 1$ are already FX-adjusted, {meaning} that they are denoted in the domestic currency according to the exchanged prices of the {actual} foreign rates. This allows for a consistent valuation under the domestic measure $\Q_0$ associated with the domestic interest and collateral rate $r_0$.
\par
Execution of the optimal collateral choice strategy imposes that the collateral is posted in the cheapest-to-deliver currency at all times, which results in discounting with the maximal collateral rate. Then, the value of the asset $U$ {with the collateral choice option} is given by
\begin{equation}\label{eq:ratecoll}
U(t) = \E^{\Q_0}_{t}\Bigl[\e^{-\int_{t}^T \max(r_0(s), \dots, r_N(s))\d s} U(T)\Bigr],\quad t_0 \leq t \leq T.
\end{equation}
\par
In the following {subsection} \ref{sec:df}, we examine the valuation formula \eqref{eq:ratecoll} closely and give conditions under which this expression can be simplified so that the impact of the collateral choice option is separated from the asset valuation. 
{\Cref{sec:stochmodel} introduces} a stochastic model associated {with} the collateral rates and {summarizes} the collateral choice valuation method proposed in \cite{WolfCF}, which will be used throughout this article.
In \Cref{sec:detctd}, we introduce the deterministic model for the collateral choice option and show how it relates to the stochastic model{, focussing} on the valuation differences. 
{In} \Cref{sec:CTD-sensi}, we compare sensitivities of the stochastic and deterministic models and demonstrate that the stochastic model {can} incorporate dependencies on multiple currencies simultaneously, a feat that the deterministic model does not reproduce.
In \Cref{sec:CTDhedging}, we consider hedging under the collateral choice option when hedging instruments are not {equipped} with this option themselves. We first show the collateral choice option's influence on collateral management through its impact on asset valuation with a synthetic replication of the collateral choice discount factor. Then, we consider the hedging of a collateralized asset with portfolios of domestic and foreign zero-coupon bonds. Strategies based on the deterministic and stochastic {models} are proposed. In \Cref{sec:CTDhedgingexperiment}, we demonstrate the hedging portfolios in a numerical experiment and {we conclude the article in \Cref{sec:conclu}}.
%
\subsection{The cheapest-to-deliver discount factor}\label{sec:df}
Formula \eqref{eq:ratecoll} for valuation of an asset with the collateral choice option is asset-specific, which renders it difficult to estimate the value of an entire book of assets {equipped} with the collateral choice option as the option would need to be computed individually for every single asset at an immense computational cost. In the following, we give assumptions under which the value of the collateral choice option can be detached from the asset and instead be expressed as a discount factor. 
Obtaining the option value in this way is highly relevant to practitioners because the separation of collateral choice option and asset price facilitates the computation of single discount factors, which can be applied to large numbers of assets without incurring the high computational cost of an asset-specific valuation.
\par
The collateral choice is often expressed in terms of the \emph{collateral basis}. In this formulation, the option is given in terms of \emph{collateral spreads} $q_i$, $i\in\{1, \dots, N\}$ which are the differences between the FX-adjusted collateral rates and the domestic rate, 
\begin{equation}\label{eq:spreaddef}
q_i(t) = r_i(t) - r_0(t).
\end{equation}
Then, the asset value \eqref{eq:ratecoll} is defined equivalently by
\begin{equation}\label{eq:spreadcollU}
U(t) = \E^{\Q_0}_{t}\Bigl[\e^{-\int_{t}^T r_0(s) \d s} \e^{-\int_{t}^T \max(0, q_1(s) \dots, q_N(s))\d s} U(T)\Bigr],
\end{equation}
since the involved maximum function can be written as
\begin{align}
\max(r_0(s), r_1(s), \dots, r_N(s)) 
&= r_0(s) + \max(0, q_1(s), \dots, q_N(s)).
\end{align}
In this general form \eqref{eq:spreadcollU}, the valuation of the asset price $U(t)$ with the collateral choice option depends on multiple stochastic processes and thus their joint distributions. 
In the following, we show how specific assumptions about the dependency structure between the domestic rate $r_0$, the collateral spreads $q_1, \dots, q_N$ and the asset $U$ itself can simplify the expression in \eqref{eq:spreadcollU}.
\par
A deterministic model of the collateral spreads $q_i$ yields a straightforward evaluation. For clarity, we denote the spreads by $\hat q_i$, $i\in\{1,\dots,N\}$ when they are modelled deterministically. Then, the expression in \eqref{eq:spreadcollU} becomes
\begin{equation}\label{eq:detcoll}
U(t) \approx \e^{-\int_{t}^T \max(0, \hat q_1(s) \dots, \hat q_N(s))\d s}\, \E^{\Q_0}_{t}\Bigl[\e^{-\int_{t}^T r_0(s) \d s} U(T)\Bigr],
\end{equation}
which is the ``standard'' risk-free valuation of $U$ with a deterministic factor to account for the collateral choice option. This assumption is valued by many practitioners for its simplicity. However, this comes at the cost of multiple drawbacks which we will highlight within this article. Besides the inherently strong assumption of no stochasticity in the collateral spreads, the deterministic assumption systematically underestimates the impact of the collateral choice option as we will show in Section~\ref{sec:detctd}. 
\par
If stochastic dynamics are to be assumed for the collateral spreads, a strong assumption is \emph{full independence} between the asset, the collateral spreads and the domestic rate. Then, \eqref{eq:spreadcollU} can be factored into three parts:
\begin{equation}\label{eq:fullindepcoll}
U(t) \approx \E^{Q_0}_{t}\Bigl[ \e^{-\int_{t}^T \max(0, q_1(s) \dots, q_N(s))\d s}\Bigr] \E^{Q_0}_{t}\Bigl[\e^{-\int_{t}^T r_0(s) \d s}\Bigr] \E^{Q_0}_{t}\Bigl[U(T)\Bigr].
\end{equation}
This approach offers a {complete} separation between the involved quantities, but independence between the domestic rate $r_0$ and the asset $U$ goes beyond standard valuation assumptions. Particularly for interest rate products, which are often among the assets equipped with the collateral choice option, this independence {does not hold}.
\par
If we allow for dependence between domestic rate $r_0$ and asset $U$ but retain independence between collateral spreads $q_i$ and the asset $U$, as well as between collateral spreads and domestic rate, we obtain the following model in which the term {that refers} to the collateral choice option can be separated:
\begin{equation}\label{eq:Qfactorization}
U(t) \approx \E^{Q_0}_{t}\Bigl[ \e^{-\int_{t}^T \max(0, q_1(s) \dots, q_N(s))\d s}\Bigr] \E^{Q_0}_{t}\Bigl[\e^{-\int_{t}^T r_0(s) \d s} U(T)\Bigr].  
\end{equation}
This form is particularly appealing since it consists of the standard valuation formula for the asset and one factor exclusively dedicated to the collateral choice option. 
\par
However, independence between the domestic rate $r_0$ and the collateral spreads $q_i$ is not required to obtain a valuation formula in the desired shape with a discount factor term. Indeed, we can slightly relax assumptions to only assume independence between the collateral spreads $q_i$, $i\geq1$ and the asset $U$ under the $T$-forward measure associated with $\Q_0$. Under this forward measure, which we denote by $\Q_0^T$, the domestic zero coupon bond $P(t, T) = \E^{\Q_0}[\exp(-\int_t^T r_0(s)\d s)]$ is the num\'eraire and by applying a change of measure to \eqref{eq:spreadcollU} we obtain
\begin{align}
U(t) 
&= P(t,T) \E^{\Q_0^T}_{t}\Bigl[ \e^{-\int_{t}^T \max(0, q_1(s) \dots, q_N(s))\d s} U(T)\Bigr].
\end{align}
Then, by the above independence assumption and returning to the original $\Q_0$ measure, the expression factors into
\begin{align}
U(t) &\approx \E^{\Q_0^T}_{t}\Bigl[ \e^{-\int_{t}^T \max(0, q_1(s) \dots, q_N(s))\d s} \Bigr] P(t,T) \E^{T}_{t}\Bigl[U(T)\Bigr] \nonumber\\
&= \E^{\Q_0^T}_{t}\Bigl[ \e^{-\int_{t}^T \max(0, q_1(s) \dots, q_N(s))\d s}\Bigr] \E^{\Q_0}_{t}\Bigl[\e^{-\int_{t}^T r_0(s)\d s} U(T)\Bigr]. \label{eq:Tfactorization}
\end{align}
This approach has the weakest assumptions of the approaches given. However, its computation can be cumbersome in the presence of multiple cash flows at times $T_k$, $k=1,\dots, m$. In this case, for every cash flow, the collateral spreads need to be modelled under the appropriate forward measures $\Q_0^{T_k}$. 
\par
We note that none of these models assumes independence between the collateral spreads themselves. That is, for spreads $q_i$ and $q_j$, $i\neq j$, 
independence between $q_i$ and $q_j$ is not required. 
\par
In all these simplified approaches, the impact of the collateral choice option can be expressed through an isolated term in the pricing equation,
\begin{equation}
U(t) \approx \CTD(t, T) \E^{\Q_0}_t\Bigl[\e^{-\int_{t}^T r_0(s)\d s} U(T)\Bigr].
\end{equation}
We denote this \emph{cheapest-to-deliver discount factor} (CTD discount factor) by
\begin{equation}\label{eq:CTDfactor}
\mathrm{CTD}(t, T) = \E_{t}\Bigl[ \e^{-\int_{t}^T \max(0, q_1(s) \dots, q_N(s))\d s}\Bigr],
\end{equation}
where the expectation is taken either under the $\Q_0$ measure (in case of independence between collateral spreads and domestic rate) 
or under the $T$-forward measure (in case of correlations between collateral spreads and domestic rate). 
In the deterministic model, this factor is simply 
\begin{equation}\label{eq:ctddet1}
\CTD_\det(t, T) = \exp\Bigl({-\int_{t}^T \max(0, \hat q_1(s) \dots, \hat q_N(s))\d s}\Bigr).
\end{equation}
In the next two sections, we introduce the respective stochastic and deterministic models used in the remainder of this article.
\section{The CTD discount factor with stochastic dynamics}\label{sec:stochmodel}
The CTD discount factor with stochastic dynamics in \eqref{eq:CTDfactor} is notoriously difficult to evaluate and requires approximations. 
Typically, the collateral spreads $q_1, \dots, q_N$ {are modelled with dynamics {analogous to stochastic interest rates}}. Under these dynamics, an analytical solution for the CTD discount factor remains unavailable; hence an approximation {must} be applied. In this article, the \emph{second-order common factor model} introduced in \cite{WolfCF} is used and briefly described below. Other approaches for the approximation step exist in the literature. In \cite{PiterbargCooking}, a {first-order} approximation and the collateral spread formulation are introduced. In \cite{SankovichZhu}, a model based on third-order moment matching for collateral rates is developed. For the {particular} case of exactly two collateral currencies, {both} a second-order model, and a model utilizing conditional independence with respect to the time axis, are given in \cite{PiterbargAntonov}.
\subsection{Stochastic collateral spread dynamics}\label{subsec:HW}
We begin by attributing stochastic dynamics to the collateral spreads. Given the inherent uncertainty in modelling their behaviour, a stochastic approach appears natural. Here, we use a one-factor Hull--White (HW) model which is commonly used for interest-rate-related processes. From initial time $t_0$ to maturity $T$, the prescribed dynamics for each collateral spread $q_i$, $i=1,\dots,N$, are
\begin{equation}\label{eq:HWspread}
\d q_i(t) = \kappa_i \left(\theta_i(t) - q_i(t)\right)\d t + \xi_i \d W_i(t), \quad t \in (t_0, T], \quad q_i(t_0)\in \R, 
\end{equation}
where $\kappa_i>0$ is the speed of mean reversion, $\xi_i>0$ the volatility coefficient and $\theta_i(t) \in \R$, $t\in(t_0, T]$ is the long-term mean by which it can be ensured that the expectations of the spreads fit market observations, $\E[q_i(t)] = \hat q_i(t)$. The driving processes $W_i$ are Brownian motions in the measure under which the CTD discount factor is taken, as specified in Section~\ref{sec:df}. That is, if the spreads are modelled with correlations to the domestic rate $r_0$, then their dynamics are modelled directly under the forward measure $\Q_0^T$. Otherwise, if the spreads are independent of the domestic rate, their dynamics are modelled under the $\Q_0$ measure. 
Correlations between the collateral spreads are defined through the instantaneous correlations $\d[W_i, W_j]_t = \rho_{i,j}\dt$, $i,j\in\{1,\dots,N\}$ between the driving Brownian motions.
{
\begin{rem}[Dynamics of the domestic interest rate]
The collateral spread dynamics are defined under a measure where independence from the domestic interest rate $r_0$ holds. Nevertheless, it is helpful to also give dynamics of the domestic rate under $\Q_0$. 
We model the domestic interest rate $r_0$ with one factor Hull--White dynamics analogous to \eqref{eq:HWspread} by
\begin{equation}\label{eq:r0dynamics}
\d r_0(t) = \kappa_0 \bigl(\theta_0(t) - r_0(t)\bigr)\d t + \xi_0 \d W_0(t), \quad t \in (t_0, T], \quad r_0(t_0)\in \R,
\end{equation}
with parameters $\kappa_0, \xi_0 > 0$, long-term mean function $\theta_0(t)\in \R$ for $t\in(t_0, T]$, and instantaneous correlations with the driving Brownian motions of the collateral spreads, $\d[W_0, W_i]_t = \rho_{0, i}\d t$ for $1\leq i\leq N$. 
\end{rem}
}
\subsection{Second-order diffusion-based common factor approximation}
In the following, we summarize the second-order diffusion-based common factor approximation introduced in \cite{WolfCF}; this model will be used in the numerical experiments of this article.
We denote the maximum process by $M(t) = \max(0, q_1(t), \dots, q_N(t))$ so that the CTD discount factor \eqref{eq:CTDfactor} can be shortly expressed by $\E[\exp(-\int_{t_0}^T M(t)\d t)]$\footnote{For notational convenience, we omit explicitly conditioning on the filtration at time $t_0$ when no other conditioning is indicated.}. No analytical solution {is} known for this term when the spreads are equipped with stochastic dynamics and correlated to another. 
The issue arises because neither the marginal distributions $M(t)$, $t\in(t_0, T]$, nor the process distributions $(M(s), M(t))$, $s,t\in(t_0, T]$ are available in closed form and thus the distribution of the integral $Y(T) := \int_{t_0}^T M(t)\d t$ is intractable.
\par
By considering the second-order Taylor expansion of $\E[\exp(-Y(T))]$ instead, we alleviate the task to only require the first two moments of the distribution of $Y(T)$. We thus approximate
\begin{align}
\E\Bigl[\exp\Bigl(-Y(T) \Bigr)\Bigr] &\approx \exp\Bigl( -\E\Bigl[ Y(T)\Bigr]\Bigr) \Biggl( 1 + \frac12 \Var\Bigl[Y(T)\Bigr]\Biggr). 
\end{align}
The task is thus to compute the moments of the integral, $\E[Y(T)]$ and $\Var[Y(T)]$. In the next step, we show how they can be approximated based on the respective moments of the maximum process, $\E[M(t)]$ and $\Var[M(t)]$. Afterwards, we give a common-factor approximation scheme which results in a semi-analytical form for these moments of the maximum process.
\par
The first moment of the integral $E[Y(T)]$ is easily expressed in terms of the first moment of the maximum process $E[M(t)]$, as Fubini's theorem yields $\E[Y(T)] = \int_{t_0}^T \E[M(t)]\d t$.
\par
The variance of the integral term $\Var[Y(T)]$, however, depends on the covariance structure $\E[M(t)M(s)]$ for $s\neq t$, which has a path-dependence that prohibits an analytical expression. By defining a standard It\^o diffusion $X(t)$ on $[t_0, T]$ such that $\Var[X(t)] = \Var[M(t)]$ for all $t\in[t_0, T]$,  we obtain an estimator for the variance of the integral,
\begin{equation}
\Var[Y(T)] = \Var[\int_{t_0}^T M(s)\d s] \approx \Var[\int_{t_0}^T X(s)\d s].
\end{equation}
This estimator, the variance of the integral of an It\^o diffusion, is analytically available.
\begin{lem}\label{lem:intvar}
The variance of the integral of $X$ is given by
\begin{equation}
\Var\Bigl[\int_{t_0}^T X(t) \d t\Bigr] = \int_{t_0}^T \int_{t_0}^s \Var[M(t)] \d t \d s + \int_{t_0}^T (T - s) \Var[M(s)] \d s.
\end{equation}
\end{lem}
We denote this diffusion-based estimator by 
\begin{equation}
\Psi(t_0, T) := \Var\Bigl[\int_{t_0}^T X(t)\d t\Bigr] \label{eq:diffusionPsi}
\end{equation} 
and by replacing $\Var[Y(T)]$ with the diffusion-based estimator, we arrive at the approximation $\CTD(t_0, T) \approx \exp(-\int_{t_0}^T \E[M(t)]\d t)(1+\frac12 \Psi(t_0, T))$. 
The proof of \Cref{lem:intvar} is given in \cite{WolfCF}, where additionally another variance of the integral estimator based on a mean-reverting process is considered.
\par
It remains to find an expression for the moments of the maximum, $\E[M(t)]$ and $\Var[M(t)]$. 
At any fixed time $t$, the HW dynamics \eqref{eq:HWspread} imply that the collateral spread vector follows a multivariate Gaussian distribution,
\begin{equation}\label{eq:multivargausscf}
{q}(t) =  [q_1(t), \dots, q_N(t)]^T \sim \mathcal{N}\Bigl([\mu_1(t), \dots, \mu_N(t)]^T, \Sigma(t)\Bigr),
\end{equation}
where the marginal distributions $q_i(t) \sim \mathcal{N}(\mu_i(t), \sigma_i^2(t))$ and the covariance matrix $\Sigma(t)$ 
are explicitly known.
We define a common factor approximation of these random vectors for times within a time discretization $\mathcal{T}$ of $[t_0, T]$. For this approximation, the required moments of the corresponding maximum process can be directly computed.
\begin{dfn}[Common factor approximation]
Let $t\in[t_0, T]$ be fixed and let $q_i(t) \sim \mathcal{N}(\mu_i(t), \sigma_i^2(t))$, $i\in\{1,\dots,N\}$ be normally distributed random variables. We define the common factor approximation 
component-wise by
\begin{equation}
\widetilde q_i(t) := C(t) + A_i(t), \quad i\in \{1,\dots,N\},
\end{equation}
where $C(t),\ A_1(t),\dots,A_N(t)$ are independent normal random variables with distributions
\begin{align}
C(t) \sim \mathcal{N}(0, \sigma^2_{\min}(t) |\gamma(t)|), \\
A_i(t) \sim \mathcal{N}(\mu_i(t), \sigma_i^2(t) - \Var[C(t)]).
\end{align}
Here, we denote the minimal variance occurring amongst the spreads $q_i(t)$ by  $\sigma_{\min}^2(t) := \min(\sigma_1^2(t), \dots, \sigma_N^2(t))$ and  determine the correlation structure of the common factor approximation by the parameter $\gamma(t)\in[0,1)$.
\end{dfn}
It can be shown that the marginal distributions coincide, $\widetilde q_i(t) \sim q_i(t)$, and the correlation structure of the common factor approximation is imposed through the parameter $\gamma(t)$ by
\begin{equation}
\corr[\widetilde q_i(t), \widetilde q_j(t)] = \frac{\sigma_{\min}^2(t) |\gamma(t)|}{\sigma_i(t)\sigma_j(t)}, \quad i \neq j \in \{1, \dots, N\}.
\end{equation}
In the three-currency case, the parameter $\gamma(t)$ is analytically available (see \cite{WolfCF}), for a higher number of currencies, it becomes an optimization problem to choose a parameter $\gamma(t)$ such that the covariance matrix of the common factor approximation is as close as possible to the covariance matrix $\Sigma(t)$ of the spreads.
\par
The common factor approximation gives rise to the common factor maximum, given by $\widetilde M(t) := \max(0, \widetilde q_1(t), \dots, \widetilde q_N(t))$. It can be equivalently expressed as
\begin{equation}
\widetilde M(t) = \max\bigl(0, C(t) + \max\bigl(A_1(t), \dots, A_N(t)\bigr)\bigr),
\end{equation}
where the inner maximum is taken between independent random variables. This construction as a sum and maximum of independent random variables makes the cumulative distribution function (cdf) of $\widetilde M(t)$ analytically available.
\begin{lem}
The cumulative distribution function of $\widetilde M(t)$ is given by
\begin{equation}
F_{\widetilde M(t)}(x) := \P[\widetilde M(t) \leq x] = 
\begin{cases}
0, & x \leq 0, \\
\Bigl(f_{C(t)} \ast \prod\limits_{i=1}^N F_{A_i(t)}\Bigr)(x), & x > 0,
\end{cases}
\end{equation}
with $f_{C(t)}(x) = \frac{1}{\sqrt{|\gamma(t)|\sigma_{\min}^2(t)}}\phi( \frac{x}{\sqrt{|\gamma(t)|\sigma_{\min}^2(t)}})$ the density of $C(t)$ and $F_{A_i(t)}(x) = \Phi(a)$ the cdf of $A_i(t)$, where 
\begin{equation}
a = \frac{x - \mu_i(t)}{\sqrt{\sigma_i^2(t) - \sigma_{\min}^2(t)|\gamma(t)|}}.
\end{equation} Here, $\varphi$ and $\Phi$ are the density and cdf of the standard normal $\mathcal{N}(0,1)$ distribution and $f*F$ denotes the convolution $(f*F)(x) = \int_\R f(z)F(x-z)\d z$. 
\end{lem}
The proof of this lemma can be found in \cite{WolfCF}.
With the cdf of the common factor maximum $\widetilde M(t)$ at hand, arbitrary moments of the distribution can be computed and it holds that
\begin{align}
\E[\widetilde M(t)] = \int_{0}^\infty 1 - F_{\widetilde M(t)}(x)\d x ,\quad
\E[\widetilde M(t)^2] = \int_0^\infty 2x \bigl(1 -  F_{\widetilde M(t)}(x)\bigr) \d x ,
\end{align}
and hence $\Var[\widetilde M(t)]$ can be obtained. Using these moments of $\widetilde M(t)$ as estimators for $\E[M(t)]$ and $\Var[M(t)]$, we define the second-order diffusion-based common factor approximation of the cheapest-to-deliver discount factor as
\begin{equation}\label{eq:CFCTD}
\CTD_\CF(t, T) := \exp\left(-\int_{t}^T \E_{t}\left[\widetilde M(s) \right]\d s\right) \left( 1 + \frac12 \Psi(t, T)\right).
\end{equation}
where $\Psi(t,T)$  is the diffusion-based variance estimator defined in \eqref{eq:diffusionPsi}, with the diffusion $X$ calibrated to the common factor maximum variance $\Var[\widetilde M(t)]$.
Proofs for all preceding lemmata, a detailed derivation and accuracy results for this stochastic model can be found in \cite{WolfCF}.
\section{The deterministic model and its connection to the stochastic model}\label{sec:detctd}
%
The deterministic model is based on market data $q_i^{\mathrm{mkt}}(t_k)$, $i\in\{1,\dots,N\}$ which is available at discrete times $t_k \in \mathcal{T}$. By suitable interpolation between these points, we obtain a continuous time function $\hat q_i(t)$, $t\in[t_0, T]$ which corresponds to each individual, deterministic collateral spread. This yields the deterministic CTD discount factor given in \eqref{eq:ctddet1}:
\begin{equation}
\CTD_\det(t, T) = \exp\Bigl({-\int_{t}^T \max(0, \hat q_1(s) \dots, \hat q_N(s))\d s}\Bigr).
\end{equation}
This describes {a unique} point of departure for the stochastic model given in \Cref{subsec:HW}. Unlike classic interest rate models, which meet zero-coupon bond prices obtained from the market, the point of departure {is} deterministic collateral spreads. A typical requirement is {a} consistency in the sense that the stochastic collateral spreads should collapse back to the deterministic collateral spreads as volatility tends to zero,
\begin{equation}\label{eq:stochtodetspread}
q_i(t) \overset{\xi_i \to 0}\longrightarrow \hat q_i(t),
\end{equation}
for all $t\in[t_0,T]$. This assumption is equivalent to setting $\E[q_i(t)] = \hat q_i(t)$ and ensures that both models perfectly fit available market data at the monitoring dates,
\begin{equation}
\E[q_i(t_k)] = \hat q_i (t_k) = q_i^{\mathrm{mkt}}(t_k).
\end{equation}
As always, the expectation is taken under the appropriate measure, $\Q_0$ or $\Q_0^T$, according to the model assumptions. 
\begin{rem}[Hull--White model calibration]\label{rem:HW-decomp}
Because of the desired connection to the deterministic curves $\hat q_i$, $i\in\{1,\dots,N\}$, it is advantageous to define the Hull--White processes $q_i$ in terms of the Hull--White decomposition, 
\begin{equation}
q_i(t) = \hat q_i(t) + u_i(t), \quad t\in[t_0, T],
\end{equation}
where $u_i(t)$ is a centred Ornstein-Uhlenbeck process with dynamics
\begin{equation}
\d u_i(t) = -\kappa_i u_i(t) \d t + \xi_i W_i(t), \quad u_i(t_0) = 0.
\end{equation}
\end{rem}
For completeness, we now give the connection between $\hat q_i(t)$ and the long-term mean function $\theta_i(t)$ from the process formulation in \eqref{eq:HWspread}. Analogous computations have been undertaken for interest rate processes; see, for example, \cite{brigo2006interest}.
In the following lemma, we give two solutions for $\theta_i(t)$ which satisfy the convergence requirement \eqref{eq:stochtodetspread}. First, a general solution in continuous time comes at the cost of requiring a differentiable interpolation of the deterministic collateral spreads. Secondly, a solution is given in which the requirement only holds at the discrete dates $t_k\in\mathcal{T}$. This is particularly suitable for numerical implementations, and for market data {that} may exhibit discontinuities.
\begin{lem}[Long-term mean function of the collateral spreads]\label{lem:thetalemma}
Let $q_i$ be the collateral spread with dynamics given in \eqref{eq:HWspread} and let its initial value be $q_i(t_0) = \hat q_i(t_0)$.
\begin{enumerate}
\item Let $\hat q_i(t)$ be a continuously differentiable function on $[t_0, T]$. Then, the choice of the long-term mean 
\begin{equation}\label{eq:newexacttheta}
\theta_i(t) = \hat q_i(t) + \frac\partial{\partial t} \frac1{\kappa_i} \hat q_i(t)
\end{equation}
satisfies $\E[q_i(t)] = \hat q_i(t)$ for every $t \in [t_0, T]$.

\item Let $\hat q_i(t_k)\in\mathbb{R}$ be defined at every time $t_k\in\mathcal{T}$ of a discretization of the time domain $[t_0, T]$.
Then, the choice of the long-term mean
\begin{equation}\label{eq:newtheta}
\theta_i(t)  = \frac{\hat q_i(t_k) \e^{\kappa_i t_k} - \hat q_i(t_{k-1}) e^{\kappa_i t_{k-1}}}{\e^{\kappa t_k} - \e^{\kappa t_{k-1}}},
\end{equation}
for $t \in (t_{k-1}, t_k]$, $k\geq 1$, 
yields that $\E[q_i(t_k)] = \hat q_i(t_k)$ holds for every $t_k \in \mathcal{T}$.
\end{enumerate}
\end{lem}
A proof of this result is given in \ref{appx:thetalemma}.
\subsection{Collateral disputes under the collateral choice option}
Unfortunately, such consistency on the level of individual spreads does not result in consistent pricing of the collateral choice option. Instead, we show that the resulting prices systematically differ.
Under the deterministic model, the cheapest-to-deliver discount factor is the purely analytical expression given in \eqref{eq:ctddet1}. By the condition \eqref{eq:stochtodetspread}, it holds
\begin{align}
\CTD_\det({t_0}, T) &= \exp\Bigl({-\int_{t_0}^T \max(0, \hat q_1(t) \dots, \hat q_N(t))\d t}\Bigr) \nonumber \\
& = \exp\Biggl({-\int_{t_0}^T \max\left(0, \E[q_1(t)], \dots, \E[q_N(t)]\right)\d t}\Biggr). \label{eq:detCTD}
\end{align}
A comparison to the generalized CTD discount factor in \eqref{eq:CTDfactor} shows that contrasted with the stochastic formulation, the deterministic version is obtained through the exchange of the expectation operator with the exponential, integral and maximum operators. This corresponds to two applications of Jensen's inequality and one application of Fubini's theorem, resulting in the inequalities
\begin{align}
\CTD({t_0}, T) & = \E\Bigl[\exp\Bigl(-\int_{t_0}^T \max\left(0, q_1(t), \dots, q_N(t)\right)\d t\Bigr)\Bigr] \nonumber \\
&\leq \exp\Bigl(-\E\Bigl[\int_{t_0}^T \max(0, q_1(t), \dots, q_N(t))\d t\Bigr]\Bigr) \nonumber \\
&\leq \exp\Bigl(-\int_{t_0}^T \max(0, \E[q_1(t)], \dots, \E[q_N(t)])\d t\Bigr) = \CTD_\det({t_0},T). \label{eq:JensenDet}
\end{align}
{The} discount factor of the deterministic model acts as an upper bound for the discount factor associated {with} the stochastic model. 
Equality is only attained when the collateral spreads are deterministic, which does not align with observations by practitioners.
Consequently, the discount factor obtained from a stochastic model is smaller than the one obtained from a deterministic model, which corresponds to the attribution of {a more significant} collateral choice option under stochastic assumptions.
\par
In the following, we will exemplify this on the concrete example of an interest rate swap.
%
The inherent valuation differences between the stochastic and the deterministic collateral choice model have an immediate impact on collateral management. We illustrate this with the example of a (payer) interest rate swap $V$ with payment dates $T_k$, $k\in \{1, \dots, m\}$, at which a fixed rate $K$ is exchanged for a simple compounded Ibor\footnote{The interest rate swaps here do not depend on the choice of floating rate. A (simple compounded forward) Ibor rate is used for illustrative purposes, but in light of the Libor transition, any other floating rate can easily replace it.} rate accrued on a notional $\bar N$. 

Again, $P(t, T_k) = \E^{\Q_0}[\exp(-\int_t^{T_k} r_0(s)\d s)]$ denotes the domestic zero coupon bond. 
The price of the swap $V$ at time $t$ is well-known to be 
\begin{align}\label{eq:noctdswap}
V(t) = \bar N \sum_{k=1}^m \tau_k P(t, T_k) \bigl(\ell_k(t) - K\bigr), 
\end{align}
where $\tau_k := T_k - T_{k-1}$ is the time period between payment dates and $\ell_k(t)$ the forward Ibor rate
\begin{equation}
\ell_k(t) := \frac{ P(t, T_{k-1}) - P(t, T_{k})}{(T_k - T_{k-1}) P(t, T_k)}.
\end{equation}
To add the (perfectly exercised, free substitution) collateral choice option, it is helpful to decompose the swap $V$ into a sum of forward rate agreements $U_k$, each with a single cash flow at time $T_k$. 
The swap equipped with the collateral choice option, denoted by $V^c$ to indicate the option, is obtained from an application of the appropriate CTD discount factor \eqref{eq:CTDfactor} to each asset $U_k$:
\begin{align}\label{eq:CTDswap}
V^c(t) &= \sum_{k=1}^m \CTD(t, T_k) \, U_k^c(t, T_k) = \bar N \sum_{k=1}^m \CTD(t, T_k) \tau_k P(t, T_k) \bigl(\ell_k(t) - K\bigr).
\end{align}
\par
In practice, the CSA agreement indicates a party (either involved in the trade or even a designated third party) {determining} the due collateral based on their asset valuation. Both trading parties involved have the right to dispute this valuation, see, for example, \cite{simmons2018collateral}. As shown above, the CTD discount factors $\CTD(t, T_k)$ have a model-dependent price, and as they are not traded assets, there is no canonical value attached. 
The inequality in \eqref{eq:JensenDet} shows that the deterministic model used by many practitioners does not adjust for a loss of convexity relative to the stochastic model. This implies that two parties with {different} valuation models of the collateral choice option are likely to arrive at different present values of their trade.
%
%
%
%
\section{Sensitivities of the collateral choice option}\label{sec:CTD-sensi}
\begin{figure}[]
\centering
\includegraphics[width=.49\textwidth, height=0.40833\textwidth]{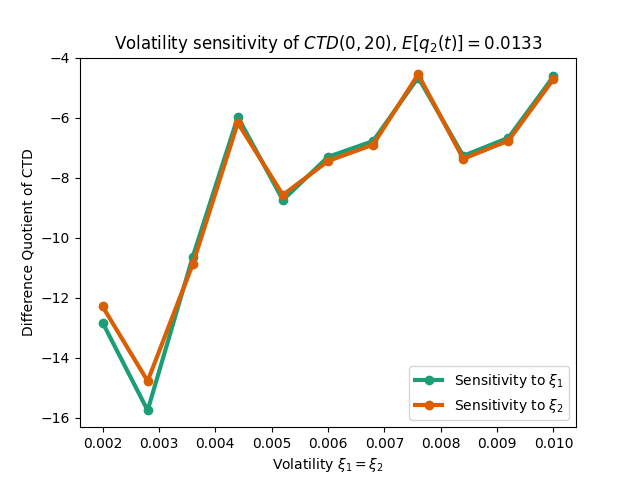}
\includegraphics[width=.49\textwidth, height=0.40833\textwidth]{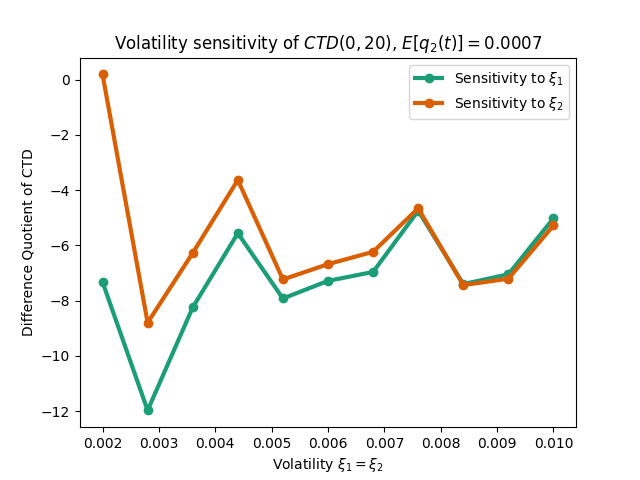}
\caption{Sensitivity of the CTD discount factor to the volatility parameters $\xi_1$ and $\xi_2$ obtained from the difference quotient of the common factor approximation. Both $\xi_1$ and $\xi_2$ are set to the same value indicated on the horizontal axis, difference quotients are obtained one parameter at a time. The speeds of mean reversion are $\kappa_1 = 0.0078$, $\kappa_2 = 0.0076$, the instantaneous correlation is $\rho_{1,2} = 0.5$. 
For all times $t\in[0,20]$, the expectations of the spreads are equal to the initial values. In the graph on the left, they are $\E[q_1(t)] \equiv 0.014$, $\E[q_2(t)] \equiv  0.0133$, whereas in the graph on the right there is a larger difference between the expectations, $\E[q_1(t)] \equiv 0.014$, $\E[q_2(t)]\equiv 0.0007$.}
\label{fig:vola-sensi}
\end{figure}
The stochastic and deterministic models for the collateral choice option also exhibit different behaviours in terms of sensitivities, as we will demonstrate in this section. We first consider the sensitivity to volatility, which the deterministic model clearly does not exhibit, and then compare how both models behave when the long-term mean of the collateral spreads is varied.
\par
%
If the collateral spreads are assumed to have stochastic dynamics, it is implied that the collateral choice option (represented by the CTD discount factor) has a sensitivity to the collateral spread volatility. We denote $\CTD(t_0, T; \xi_i)$ the CTD discount factor obtained when the $i$-th collateral spread, for some $i\in\{1,\dots,N\}$, has volatility parameter $\xi_i$. The described sensitivity is then
\begin{equation}
\frac{\partial \CTD(t_0, T)}{\partial \xi_i} = \lim\limits_{\varepsilon \to 0} \frac{\CTD(t_0, T; \xi_i + \varepsilon) - \CTD(t_0, T; \xi_i)}{\varepsilon}.
\end{equation}
Since the stochastic CTD discount factor does not admit an analytical solution, we have to resort to a ``bump-and-revalue'' evaluation where the CTD discount factor is computed with the different parameters considered. In the following, we give an example based on the common factor approach given in \eqref{eq:CFCTD} with a two-sided finite difference scheme,
\begin{equation}
\frac{\partial \CTD(t_0, T)}{\partial \xi_i} \approx \frac{\CTD_\CF(t_0, T; \xi_i + \varepsilon) - \CTD_\CF(t_0, T; \xi_i-\varepsilon)}{2\varepsilon} .
\end{equation}
The example consists of three currencies, i.e.\ two collateral spreads and one domestic currency corresponding to the zero component in $\max(0, q_1(t), q_2(t))$.
 To improve the interpretability of the results, constant expectations are assumed for both of the collateral spreads, and the expectations are ordered such that $\E[q_1(t)] > \E[q_2(t)] > 0$. Furthermore, expectations of the spreads are exaggerated with respect to recent market observations for this example. 
The results are presented in \Cref{fig:vola-sensi}. It becomes apparent that each collateral spread impacts the CTD discount factor, not only the collateral spread with the highest expected value, which would be the uncontested, sole maximal spread under deterministic dynamics.
%
%
The size of {each collateral spread's impact} on the CTD discount factor is related to the spread's probability of being the maximal spread. 
As volatility increases, the spreads exhibit more variation in their paths, increasing the initially smaller spread's probability of being the maximal spread. Consequently, the difference between the sensitivity to either volatility parameter 
decreases {as the volatility parameters grow}.
In the graph on the right, we chose a significantly smaller initial value and expectation for the collateral spread $q_2(t)$ to reduce its probability of being the maximal spread. 
Particularly in the low-volatility scenarios, there is, heuristically speaking,  not enough random motion in the system for the two spreads to cross. 
Consequently, the sensitivity of the CTD discount factor to $\xi_2$ is reduced in these cases. We conclude that the CTD discount factor exhibits sensitivities to the collateral spreads' volatilities and if the collateral spreads are to be thought of as stochastic processes, then this is an {integral} part of the picture which the deterministic model, and hedging strategies derived from it, cannot take into account.
\par
%
\begin{figure}[]
\centering
\includegraphics[width=.49\textwidth, height=0.40833\textwidth]{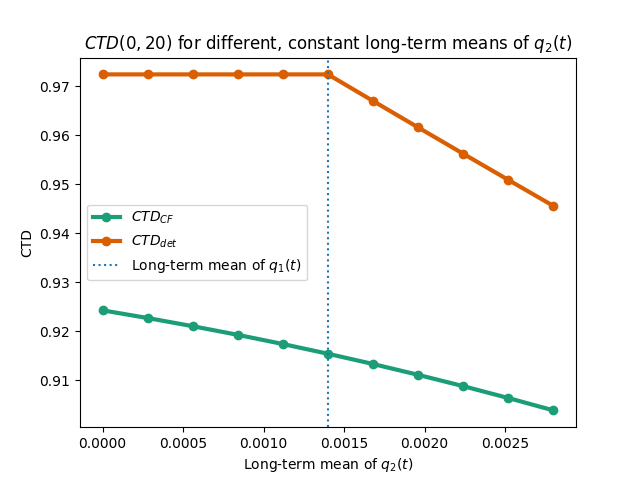}
\includegraphics[width=.49\textwidth, height=0.40833\textwidth]{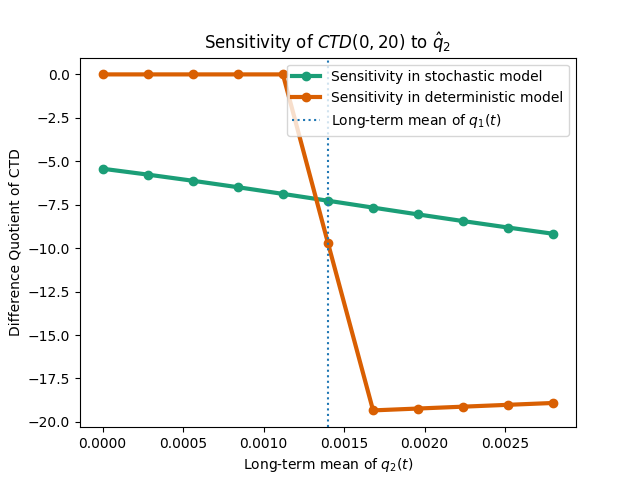}
\caption{Sensitivity of the stochastic and the deterministic CTD discount factor to changes in the long-term mean of $q_2(t)$. Long-term means $\hat q_1$ and $\hat q_2$ are constant at the values indicated in the figure. Volatilities are $\xi_1 = 0.0018$, $\xi_2 = 0.0023$, correlation and speed of mean reversion are as indicated in \Cref{fig:vola-sensi}. In the graph on the left, the CTD discount factors obtained from the deterministic and the common factor method are given for different values of $\hat q_2$. In the graph on the right, the difference quotient relative to changes in the long-term mean $\hat q_2$ is given.}
\label{fig:delta-sensi}
\end{figure}
When the sensitivity to changes in the long-term mean of the underlying collateral spread is considered, more differences between the stochastic and deterministic models become apparent. For illustrative purposes, we again consider the simplified scenario where both collateral spreads $q_1(t)$ and $q_2(t)$ have constant expectations.
The deterministic model is thus described by
\begin{equation}
\hat q_i(t) \equiv \hat q_i \in \R,\ i \in \{1, 2\}.
\end{equation}

As outlined in \Cref{rem:HW-decomp}, the corresponding stochastic model can be given by 
\begin{align}
q_i(t) &= \hat q_i + u_i(t), & t\in[t_0, T] \\
\d u_i(t) &= -\kappa_i u_i(t) \d t + \xi_i \d W_i(t), &  u_i(t_0) = 0,
\end{align}
for $ i\in\{1, 2\}$. One can show that this is equivalent to the model described in \eqref{eq:HWspread} with long-term mean parameters constant and equal to the initial value $\hat q_i$.
\par
We consider the sensitivity of the CTD discount factor to changes in $\hat q_2$, i.e. the (average) level of the first collateral spread remains fixed and the (average) level of the second collateral spread is varied. We again use the notation $\CTD(t_0, T; \hat q_2)$ to indicate the CTD discount factor obtained under the choice of  $\hat q_2$. Then, the following sensitivity is considered:
\begin{equation}
\frac{\partial \CTD(t_0, T)}{\partial \hat q_2} = 
\lim\limits_{\varepsilon \to 0} \frac{\CTD(t_0, T; \hat q_2 + \varepsilon) - \CTD(t_0, T; \hat q_2)}{\varepsilon}.
\end{equation}
For comparison of the stochastic and the deterministic model, we consider the two-sided finite difference schemes
\begin{align}
\frac{\partial \CTD(t_0, T)}{\partial \hat q_2} \approx \frac{\CTD_\CF(t_0, T; \hat q_2 + \varepsilon) - \CTD_\CF(t_0, T; \hat q_2- \varepsilon)}{2\varepsilon},
\end{align}
for evaluation under the stochastic model, and
\begin{align}
\frac{\partial \CTD(t_0, T)}{\partial \hat q_2} \approx \frac{\CTD_\det(t_0, T; \hat q_2 + \varepsilon) - \CTD_\det(t_0, T; \hat q_2- \varepsilon)}{2\varepsilon},
\end{align}
for evaluation under the deterministic model, respectively. The results are presented in \Cref{fig:delta-sensi}, where the graph on the left shows the resulting discount factors $\CTD_\CF(t_0, T; \hat q_2)$ and $\CTD_\det(t_0, T; \hat q_2)$ and the graph on the right shows the corresponding difference quotients.
\par
First, the left graph of \Cref{fig:delta-sensi} shows what was discussed in \Cref{sec:detctd}, the deterministic model undervalues the price of the collateral choice option. Secondly, the deterministic model does not consider changes to the spread which is not the maximal spread. {Therefore, for} values $\hat q_2 < \hat q_1$, the value of the CTD discount factor remains unchanged. The same can be observed in the graph of the sensitivity on the right. Moreover, as $\hat q_2$ becomes larger than $\hat q_1$, the deterministic model overestimates the sensitivity. This can be attributed to the ``binary'' nature of the deterministic model. Once it sees the second collateral spread as maximal, any changes applied to the second collateral spread fully impact the CTD discount factor. {Analogously, after this point the deterministic model loses sensitivity to the first collateral spread.} The stochastic model, on the other hand, exhibits a ``dampening'' behaviour, where the sensitivity impact of changes to the long-term mean $\hat q_2$ is related to the probability of $q_2(t)$ being the maximal spread. As we get to the hedging of the collateral choice option in the next section, we will see this behaviour of the two models reflected in the 
hedging portfolio.
%
%
\section{Hedging under the collateral choice option}\label{sec:CTDhedging}
The presence of a collateral choice option impacts more than the valuation of assets. From a risk management perspective, there are two implications. First, hedging strategies need to be adapted to the presence of the collateral choice option, which changes the asset valuation. Regarding this effect, we show that a precise valuation of the collateral choice option is crucial in \Cref{sec:Deltahedge}. 
{Secondly, the collateral choice option itself} introduces risk factors which should be hedged against. In \Cref{ref:ctdhedging}, we introduce strategies for the hedging of an asset {equipped} with the collateral choice model with foreign and domestic zero-coupon bonds. These hedging strategies are based on both the deterministic and the stochastic framework of modelling the collateral choice option.
\par
We only consider {cases} where the hedging instruments are not equipped with the collateral choice option. Otherwise, if the same collateral choice is available to the hedging instruments, back-to-back transactions become possible in which cash flows of the hedged asset are offset by cash flows of the hedging instruments, which neutralizes all risks associated {with} the collateral choice option and consequently dissolves the associated hedging problem.
\subsection{Adaption of the classic hedging strategy}\label{sec:Deltahedge}
The presence of the collateral choice option requires adaptions of classical hedging strategies. Exemplarily, we consider again the interest rate swaps $V^c$ and $V$ defined in \eqref{eq:CTDswap} and \eqref{eq:noctdswap}, with and without the collateral choice option, respectively.
Both swaps are linear products in the underlying bonds $P(t, T_k)$, but the swap $V^c$ has additional CTD discount factors in each summand. As always, the expectation is taken under the domestic measure $\Q_0$ or under forward measures $\Q_0^{T_k}$, depending on the modelling choices described in Section~\ref{sec:df}.
\par
The classic hedging strategy for the uncollateralized interest rate swap $V$ is given by a hedging portfolio, denoted by $\Pi(t)$, 
which consists of the swap and a linear combination of zero-coupon bonds $P(t_0, T_k)$, $T_k\in\{T_1, \dots, T_m\}$. In such a case of a linear asset, the portfolio $\Pi(t)$ 
is static, i.e.\ it can be set up once at initial time $t_0$ for the entire maturity and it perfectly hedges any price changes in the underlyings, $P(0, T_k)$. 
\par
When the collateral choice option is added, the pricing formula for the swap changes to the form given in \eqref{eq:CTDswap} and an adaption of the classic hedging strategy becomes necessary to accommodate the CTD discount factor terms. 
\par
As per our requirement, the collateral choice option is not available for the hedging instruments. Without this option, the CTD discount factors $\CTD(t, T_k)$ are difficult to replicate with market instruments, as the involved quantities, the collateral spreads, are not tradeable assets themselves. In the following, we demonstrate that a synthetic replication of the discount factor can be used to circumvent this problem.
\par
Since the value of the CTD discount factor is not directly available from market data and does not admit an analytical formula either, it remains to be approximated. In this example, we consider three hedged portfolios based on different approximation schemes: 
In the first scheme, the CTD discount factors are ignored, which corresponds to the approximation $\CTD(t, T_k) \approx 
 1$. This results in the standard hedged portfolio $\Pi(t)$ as described for the uncollateralized swap $V$ above. \par
For the second scheme, we consider the deterministic model $\CTD(t, T_k) \approx \CTD_\det(t, T_k)$ described in \Cref{sec:detctd}. Finally, for the third scheme, we consider a stochastic model with the common factor discount factors $\CTD(t, T_k) \approx \CTD_\CF(t, T_k)$ outlined in \Cref{sec:stochmodel}.
\par
Since the collateral choice option is not available to the hedging instruments, we {equip} them with synthetic discount factors which are applied to their respective notionals. We therefore construct hedging portfolios $\Pi^c_j$, $j\in\{0, 1, 2\}$ corresponding to the schemes outlined before by use of the synthetic discount factors $C_0(t, T_k) \equiv 1$, $C_1(t, T_k) = \CTD_\det(t, T_k)$ and $C_2(t, T_k) = \CTD_\CF(t, T_k)$.

The resulting hedging portfolios 
are 
\begin{equation}
\Pi^c_{j}(t) = V^c(t) - \sum_{T_k \geq t} C_j(t, T_k) \bar N\tau_k (\ell_k(t) - K)  P(t, T_k),
\end{equation}
for $j\in\{0, 1, 2\}$.
\par
Notably, $\Pi^c_0$ is the standard hedging portfolio of a swap without the collateral choice option.
The portfolios $\Pi^c_1$ and $\Pi^c_2$, which replicate the CTD discount factor, are no longer static hedging portfolios, as the synthetic replication of the discount factors requires a dynamic rescaling of the notional amounts as time progresses.
\par
We compare the performance of the hedging portfolios by the distribution of their associated P\&L accounts at final maturity $T_m$.
If it is ensured that the set of rebalancing times of the hedging portfolios, which we denote by $\mathcal{T}^\Pi$, includes the payment dates $T_k$, 
then the P\&L accounts are given by 
\begin{align}
\PnL_j(t_0) &= \Pi^c_j(t_0), \\ 
\PnL_j(t_\ell) &= \PnL_j(t_{\ell-1}) \e^{\int_{t_{\ell-1}}^{t_\ell}r_0(s)\d s} + (\Pi^c_{j}(t_{\ell}) - \Pi^c_{j}(t_{\ell-1})),
\end{align}
for all $t_\ell \in \mathcal{T}^\Pi$.
By requiring that $T_k \in \mathcal{T}^\Pi$, we guarantee that the synthetic discount factors correspond to the actual CTD discount factors, $\CTD(T_k, T_k) = 1 = C_j(T_k, T_k)$. This ensures that the realized cash flows from the swap $V^c$ at times $T_k$ are perfectly neutralized by the hedge and no residuals are accumulated in the P\&L account which would distort the results.
\par
Figure~\ref{fig:PnL-hist} exemplarily shows the distribution of the three P\&L accounts $\PnL_j(T_m)$ at final maturity for $j\in\{0,1,2\}$. In the experiment, the market is simulated with the mean-reverting stochastic dynamics for the collateral spreads $q_i,\ i\in\{1, \dots, N\}$ given in \eqref{eq:HWspread}, and the same type of dynamics for the domestic collateral rate $r_0$. From the figure, it becomes apparent that the presence of the collateral choice option should not be discarded in hedging, and that the quality of the synthetic replication of the CTD discount factor can play a significant role. 
\begin{figure}
\centering
\includegraphics[width=.7\textwidth]{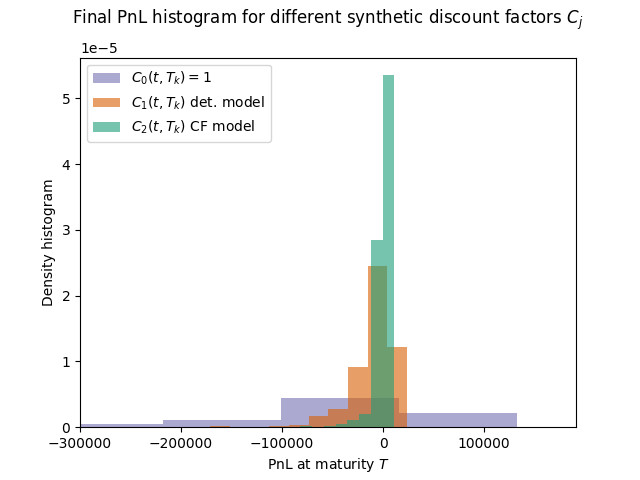}
\caption[]{Final P\&L distribution histogram of a $10$ year interest rate swap with {the} collateral choice option, hedged with synthetic replication of the cheapest-to-deliver discount factors obtained from the deterministic, {respectively stochastic, model}. The notional is $10^7$ units of the domestic currency.}
\label{fig:PnL-hist}
\end{figure}
\subsection{Hedging of the collateral choice risk}\label{ref:ctdhedging}
In the previous section, we demonstrated the immediate impact of the collateral choice option on asset valuation and its consequences in hedging. The implied strategy of synthetically replicating the CTD discount factors is not very practical, as it requires a frequent rebalancing of the hedging portfolio which is associated with transaction costs.
\par
In this section, we introduce static hedging strategies for assets with the collateral choice option, which require no rebalancing and take into consideration the risks added by the collateral choice option. As detailed before, we do not allow for the hedging instruments to carry the collateral choice option.
\par
To this end, we consider the hedging problem of a domestic zero-coupon bond which is equipped with the collateral choice option. Zero-coupon bonds form an important building block for many interest rate products, and, as we saw in \eqref{eq:CTDswap}, the CTD discount factor directly acts on the bonds. Therefore, analysing the hedging of a zero-coupon bond is the first step toward further-reaching conclusions. 
\par
 We denote the domestic zero-coupon bond equipped with the collateral choice option by $P^c$, its price is given by
\begin{equation}\label{eq:Uc}
P^c(t, T) = \E^{\Q_0}_t[\e^{-\int_t^T \max(0, q_1(s), \dots, q_N(s))\d s}\e^{-\int_t^T r_0(s)\ds}]=\CTD(t, T)P(t,T),
\end{equation}
for all $t\in[t_0, T]$. As before, $\CTD(t, T)$ denotes the CTD discount factor defined in \eqref{eq:CTDfactor} under the domestic measure $\Q_0$ or the forward measure $\Q_0^T$, according to the dependence modelling of collateral spreads and domestic interest rate.
\par
We propose to hedge the asset $P^c$ with combinations of the domestic zero-coupon bond and the foreign zero-coupon bonds associated with the collateral currencies which appear in the asset's collateral choice option. {For the sake of demonstration, we express the domestic prices of foreign zero-coupon bonds in terms of the FX-adjusted collateral rates $r_i$, $i\in\{1, \dots, N\}$, to simplify valuation under the domestic measure $\Q_0$.
\begin{dfn}[Domestic price of a foreign zero-coupon bond]\label{def:spreadbond}
A foreign zero-coupon bond pays one unit of foreign currency at maturity $T$. Its domestic price at time $t\in[t_0, T]$ is obtained by discounting under the FX-adjusted foreign interest rate $r_i(s)$, 
\begin{equation}
Q_i(t, T) = \E^{\Q_0}_t\Bigl[ \exp\bigl(-\int_t^T r_i(s)\d s\bigr) \Bigr].
\end{equation}
\end{dfn}
}
We remark that the foreign zero-coupon bond can be expressed by using the domestic discount factor and a discount factor based on the collateral spread $q_i(s)$,
\begin{equation}\label{eq:collspreadbond}
Q_i(t, T) := \E^{\Q_0}_t\Bigl[\exp\bigl(-\int_t^{T} q_i(s) \d s\bigr) \exp\bigl(-\int_t^T r_0(s)\d s\bigr)\Bigr].
\end{equation}

Analogous to the discussion in \Cref{sec:df}, we can write
\begin{equation}
Q_i(t, T) = \E^{\Q^*}_t\Bigl[\exp\bigl(-\int_t^{T} q_i(s) \d s\bigr)\Bigr] P(t, T),
\end{equation}
for $\Q^* \in\{\Q_0, \Q_0^T\}$ depending on the modelling of the collateral spreads.
\begin{rem}
With the convention that $q_0(t) = 0$, we can express the domestic zero-coupon bond $P$ in this framework as 
\begin{equation}
Q_0(t, T) = \E[\exp(-\int_t^T q_0(s) \d s)] P(t, T) = P(t, T).
\end{equation}
\end{rem}
\par
For the remainder of this section, we facilitate the setting with the assumptions of \eqref{eq:Qfactorization}, independence between the collateral spreads $q_i(t)$, $i\in\{1, \dots, N\}$ and the domestic interest rate $r_0(t)$. Unless indicated differently, all following expectations are taken under the $\Q_0$ domestic measure.
\par
First, we present a hedging strategy that aims to eliminate the risks introduced through the collateral choice option from the asset $P^c(t, T)$, based on the assumption of stochastic collateral spread dynamics laid out in \Cref{sec:stochmodel}. The underlying idea is to find weights $\alpha = (\alpha_0, \dots, \alpha_N)$ which correspond to each bond $Q_i$ 
in a way that accounts for the probabilities with which each collateral spread is the maximal spread during the lifetime of the asset. As indicated in \Cref{sec:CTD-sensi}, the CTD discount factor depends on many risk factors that should be taken into consideration.
\par
Therefore, we propose a static hedging portfolio based on variance minimization with zero initial price. 
To this end, define the stochastic strategy CTD-hedging portfolio
\begin{equation}\label{eq:Pistoch0}
\Pi_\stoch (t_0) := P^c(t_0, T) + \sum\limits_{i=0}^N \alpha_i Q_i(t_0, T) + \mathcal{C}_\stoch(t_0),
\end{equation}
where $\mathcal{C}_\stoch(t_0)$ denotes the cash account offsetting the transactions at initial time $t_0$,
\begin{align}
\mathcal{C}_\stoch(t_0) &= - \bigl(P^c(t_0, T) + \sum\limits_{i=0}^N \alpha_i Q_i(t_0, T)\bigr) \nonumber \\
&= - \Bigl( \E\Bigl[ \e^{-\int_{t_0}^T \max(0, q_1(s),\dots, q_N(s))\d s}\Bigr] + \sum\limits_{i=0}^N \alpha_i \E\Bigl[\e^{-\int_{t_0}^T q_i(s)\d s} \Bigr] \Bigr) P(t_0, T)
\end{align}
The value of the hedging portfolio at times $t_0 \leq t \leq T$ is then given by
\begin{align}\label{eq:Pistoch}
\Pi_\stoch (t) := P^c(t, T) + \sum\limits_{i=0}^N \alpha_i Q_i(t, T) + \mathcal{C}_\stoch(t_0)B(t_0, t).
\end{align}
Here, $B(t_0, t) = \exp(\int_{t_0}^t r_0(s)\d s)$ denotes the ``bank account'', which is the num\'eraire of the $\Q_0$ measure and reflects the accrual of the cash amount (positive or negative) held.
In the following, we will define a random variable $\pi(\alpha)$ related to the hedging portfolio's performance across its lifetime $[t_0, T]$ as a function of the hedging weights $\alpha$.
\begin{dfn}\label{defi:performancevar}
Let $\widetilde P^c(t_0, T)$ and $\widetilde Q_i(t_0, T)$, $i\in\{0,\dots, N\}$ be the random variables associated with the collateral choice zero-coupon bond $P^c(t_0, T)$ and the bonds $Q_i(t_0, T)$. They are given by
\begin{align}
\widetilde P^c(t_0, T) &= \exp\Bigl(-\int_{t_0}^T \max\bigl(0, q_1(s), \dots, q_N(s)\bigr)\d s\Bigr) \exp\Bigl(-\int_{t_0}^T r_0(s)\ds\Bigr)  \\
\widetilde Q_i(t_0, T) &= \exp\Bigl(-\int_{t_0}^T q_i(s) \d s\Bigr) \exp\Bigl(-\int_{t_0}^T r_0(s)\ds\Bigr),\quad 1 \leq i \leq N, \nonumber \\
\widetilde Q_0(t_0, T) &= \exp\Bigl(-\int_{t_0}^T r_0(s) \d s\Bigr).
\end{align}
Finally, let $\pi(\alpha)$ be the performance random variable of the hedging portfolio $\Pi_\stoch(t_0)$ over its lifetime, defined by
\begin{equation}\label{eq:defpialpha}
\pi(\alpha) = \widetilde P^c(t_0, T) + \sum\limits_{i=0}^N \alpha_i \widetilde Q_i(t_0, T) + \mathcal{C}_\stoch(t_0).
\end{equation}
\end{dfn}
With these ingredients at hand, we can define the portfolio variance minimization by which the weights $\alpha = (\alpha_0, \dots, \alpha_N)$ are found. 
The variance of $\pi(\alpha)$ is given by
\begin{align}
\Var[\pi(\alpha)] &= \Var\bigl[\widetilde P^c(t_0, T) + \sum\limits_{i=0}^N \alpha_i \widetilde Q_i(t_0, T)\bigr],
\end{align}
where we used that $\mathcal{C}_\stoch(t_0)$ is constant.
We find the weights $\alpha_i$, $i\in\{0,\dots, N\}$ by minimizing the variance of the random variable $\pi(\alpha)$, 
\begin{equation}\label{eq:minimizevar}
(\alpha_0, \dots, \alpha_N) = \mathrm{argmin}\Bigr\{\alpha' \in [-1, 1]^{N+1}\colon  \Var[\pi(\alpha')]\Bigr\}. 
\end{equation}
Note that the minimizing weights $\alpha$ obtained in \eqref{eq:minimizevar} are equal to the weights obtained from minimizing
%
\begin{align}\label{eq:f-varmin}
f(\alpha) = &\sum_{i=0}^N \alpha_i^2 \Var[\widetilde Q_i(t_0, T)] + \sum_{\overset{i,j=0}{j\neq i}}^N \alpha_i \alpha_j \Cov[\widetilde Q_i(t_0, T), \widetilde Q_j(t_0, T)] \nonumber \\
&+ 2 \sum_{i=0}^N \alpha_i \Cov[\widetilde P^c(t_0, T), \widetilde Q_i(t_0, T)].
\end{align}
The variance and covariance terms in $f(\alpha)$ can be expressed analytically, respectively semi-analytically with the common factor approximation for the term containing $\widetilde P^c(t_0, T)$, due to the embedded CTD discount factor. Details on this procedure are given in \ref{appx:varianceminimization}.

\par
We contrast the obtained hedging portfolio with another static hedging portfolio, based on deterministic modelling of the collateral spreads. This second portfolio ensures that the deterministic maximum spread obtained from the deterministic collateral spread forecast is matched precisely. 
\par
This is achieved by a combination of the asset $P^c(t_0, T)$ with bonds $Q_i(t_0, T)$ and forward contracts $F_i(t_0, S, T)$, which deliver bonds $Q_i(S, T)$ at times $S\in[t_0, T]$. The reasoning behind the portfolio comes from the perceived collateral spread structure in a deterministic model. The deterministic collateral spreads $\hat q_i(t)$ imply deterministic crossing times, which are the times where the spread that attains the maximum changes. Based on these, we can set up a static portfolio at time $t_0$, where the corresponding bond to the deterministic maximum is contained at every time $t_0 \leq t \leq T$.
\par
We begin with two definitions needed for the deterministic strategy. 
\begin{dfn}[Forward contract on a foreign zero-coupon bond in the domestic currency]
Let the forward contract $F_i(t, S, T)$, $i\in\{1,\dots,N\}$ deliver a {foreign zero-coupon bond denominated in the domestic currency} at time $S$ with maturity $T$. 
No-arbitrage conditions dictate that the forward contract is equivalent to 
\begin{equation}
F_i(t, S, T) = \frac{Q_i(t, T)}{P(t, S)}.
\end{equation}
Let the forward contract $F_0(t, S, T)$ analogously deliver the domestic zero coupon bond,
\begin{equation}
F_0(t, S, T) = \frac{P(t, T)}{P(t, S)}.
\end{equation}
\end{dfn}
\begin{dfn}[Maximal spread crossing time]
Let the maximum over the deterministic collateral spreads for all times $t \in [t_0, T]$ be denoted by
\begin{equation}
\hat M(t) := \max(\hat q_0(t), \hat q_1(t), \dots, \hat q_N(t)).
\end{equation}
\par
Let the crossing time $T^C_k$ be the $k$th time that the collateral spread which attains the maximum changes. 
The crossing times are defined iteratively for $k\geq 1$ by
\begin{align}
T^C_0 &:= t_0, \nonumber \\
T^C_k & := \inf\bigl\{t > T^C_{k-1}\colon \{0 \leq i \leq N \colon \hat q_i(t) = \hat M(t)\} \neq \{0 \leq j \leq N \colon \hat q_j(T^C_{k-1}) = \hat M(T^C_{k-1})\} \bigr\}.
\end{align}
Let $\mathbb{S}$ be the set of crossing times before maturity, $\mathbb{S} = \{T^C_k\colon k\geq0, T^C_k < T\}$.
This set is nonempty as $t_0 =  T^C_0 \in \mathbb{S}$.
\end{dfn}
We recall that $\hat q_0(t) = 0$ by the definition of the collateral spreads in \eqref{eq:spreaddef}.
Without loss of generality, let exactly one of the collateral spreads $\hat q_i(t)$ attain the maximum over any given time interval. For intervals where this does not hold, these maximal collateral spreads are indistinguishable for our purposes and further criteria, like preferences of the collateral posting party, can be considered to create a distinction.
\par
Let $i_k \in \{0, \dots, N\}$ denote the index of the collateral spread $\hat q_i(t)$ which is maximal on the interval $[T^C_{k}, T^C_{k+1})$ for $k\geq 0$. 
Under the deterministic CTD model, the asset is priced as
\begin{align}
\hat P^c(t, T) &= \E\Bigl[\exp\Bigl(-\int_t^T \hat M(s) \d s\Bigr) P(t, T) \Bigr] \nonumber \\
&= \E\Bigl[ \exp\Bigl(-\int_t^T \sum\limits_{T_k^C\in\mathbb{S}} \hat q_{i_k}(s) \1_{[T_k^C, T^C_{k+1})}(s) \d s\Bigr)\Bigr] P(t, T).
\end{align}
The optimal deterministic CTD-hedging strategy is then given by a portfolio which consists of the asset $P^c(t, T)$, and is short in exactly one unit of the {optimal} bond $Q_{i_k}(t, T)$ during times $t\in[T^C_{k}, T^C_{k+1})$, $k\geq 0$.
It is obtained by the following procedure:
\par
\begin{equation*}
\text{For each } T_k^C \in \mathbb{S}, \text{ add:}
\begin{cases}
-F_{i_k}(t_0, T_k^C, T) + F_{i_k}(t_0, T_{k+1}^C, T), & \text{ if } T_{k+1}^C \in \mathbb{S}, \\
-F_{i_k}(t_0, T_k^C, T), & \text{ if } T_{k+1}^C \not\in \mathbb{S},
\end{cases}
\end{equation*}
to the portfolio.
%
The static, deterministic strategy CTD-hedging portfolio with zero initial price is thus given by
\begin{equation}\label{eq:Pidet}
\Pi_\det(t) := P^c(t, T) + \sum\limits_{T^C_{k+1}\in\mathbb{S}} \Bigl(-F_{i_k}(t, T^C_{k}, T) + F_{i_k}(t, T^C_{k+1}, T)\Bigr) -  F_{i_m}(t, T^C_m, T)  + \mathcal{C}_\det(t),
\end{equation}
where $T^C_m= \max\{T^C_k \in \mathbb{S}\}$ is the final crossing time before maturity and $i_m$ the associated index of the collateral spread $\hat q_i(t)$ which is maximal on $[T^C_m, T)$.
Again, $\mathcal{C}_\det(t)$ denotes the cash account for the initial condition, which is determined by accrual with the domestic interest rate, $\mathcal{C}_\det(t) = \mathcal{C}_\det(t_0)B(t_0, t)$, with initial holdings $\mathcal{C}_\det(t_0)$ obtained from offsetting the transactions at $t_0$,
\begin{equation}
\mathcal{C}_\det(t_0) = - \Bigl(P^c(t_0, T) + \sum\limits_{T^C_{k+1}\in\mathbb{S}} \Bigl(-F_{i_k}(t_0, T^C_{k}, T) + F_{i_k}(t_0, T^C_{k+1}, T)\Bigr) -  F_{i_m}(t_0, T^C_m, T)\Bigr).
\end{equation}
In \eqref{eq:Pidet} above, we implicitly require physical settlement of the forward contracts, so that $F_i(t, T^C_k, T) = Q_i(t, T)$ for times $t>T^C_k$.
\par
Finally, we define a class of basic portfolios $\Pi_i$, $i\in\{0, \dots, N\}$ that simply hedge the collateral choice zero-coupon bond $P^c$ with one unit of the bond $Q_i$:
\begin{equation}\label{eq:Pibasic}
\Pi_i(t) = P^c(t, T) - Q_i(t, T) + \mathcal{C}_i(t).
\end{equation}
The cash accounts of the basic portfolios $\Pi_i$ are given by
\begin{equation}
\mathcal{C}_i(t) = - \bigl(P^c(t_0, T) - Q_i(t_0, T)\bigr) B(t_0, t).
\end{equation}
 We include the domestic currency choice $Q_0(t, T) := P(t, T)$ to obtain the portfolio
\begin{equation}
\Pi_0(t) = P^c(t, T) - P(t, T) + \mathcal{C}_0(t)B(t_0, t),
\end{equation}
which simply ignores the collateral choice option in the hedge and corresponds to the portfolio $\Pi^c_0$ in \Cref{sec:Deltahedge}. We emphasize this by writing $\Pi_\none := \Pi_0$.
In absence of crossing times $t_0 < T^C_k < T$, the basic portfolio of the maximal (cheapest-to-deliver) currency is equivalent to the deterministic strategy introduced before.
%
\section{Numerical experiments on CTD hedging}\label{sec:CTDhedgingexperiment}
%
We consider numerical experiments to compare the performances of the proposed hedging strategies. The ``stochastic strategy'' is based on stochastic collateral spread assumptions associated with the portfolio $\Pi_\stoch$ given in \eqref{eq:Pistoch}, the ``deterministic strategy'' is based on deterministic collateral spread assumptions associated with the portfolio $\Pi_\det$ given in \eqref{eq:Pidet} and the ``basic strategies'' correspond to the portfolios $\Pi_i$, $i\in\{0, \dots, N\}$ given in \eqref{eq:Pibasic}.
\par
We add the assumption that the domestic collateral rate is constant zero, $r_0(t) \equiv 0$. By this, the collateralized asset, defined in \eqref{eq:Uc},  simplifies to $P^c(t, T) = \CTD(t, T)$. This ensures that any risks encountered in the following 
stem exclusively from the collateral choice option. Consequently, it holds that $P(t, T) = B(t, T) \equiv 1$ for all $t_0 \leq t \leq T$.
\par
We return to the setting of the experiments in \Cref{sec:CTD-sensi} with three available currencies, i.e.\ two non-zero collateral spreads $q_1(t)$ and $q_2(t)$. The speed of mean reversion parameters are set to $\kappa_1 = 0.0078$, $\kappa_2 = 0.0076$, the volatility parameters to $\xi_1 = 0.0018$, $\xi_2 = 0.0023$ and the instantaneous correlation to $\rho_{1,2} = 0.3$. We remark that throughout this section, the deterministic collateral spreads and stochastic collateral spread expectations, $\hat q_i(t) = \E[q_i(t)]$ are modelled as linear functions. This is not a general restriction of the model but a simplifying choice to facilitate interpretations.
\par
\begin{figure}[]
\centering
\includegraphics[width=.49\textwidth, height=0.40833\textwidth]{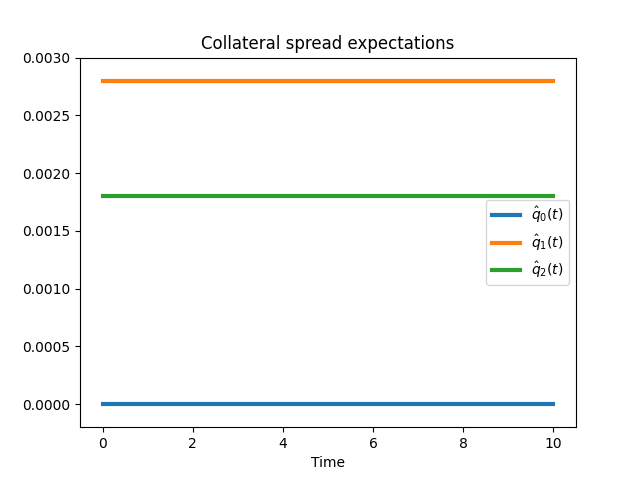}
\includegraphics[width=.49\textwidth, height=0.40833\textwidth]{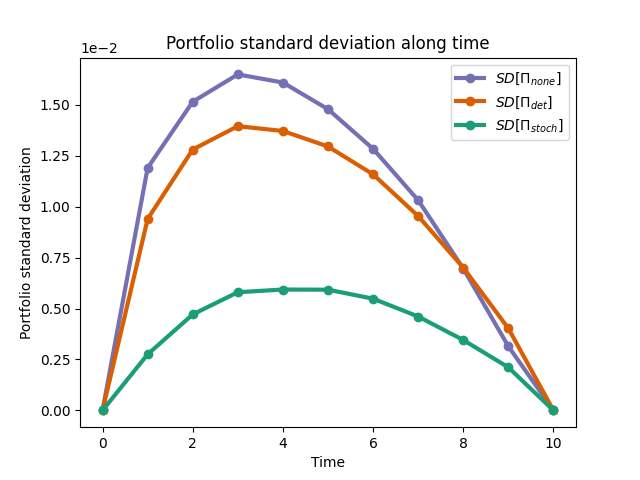}
\caption{In the first CTD-hedging experiment, the deterministic collateral spreads do not cross for the entire lifetime of the asset, $t\in[0, 10]$. In the graph on the left, the expectations of the collateral spreads are given. In the graph on the right, we give the standard deviations of the CTD-hedging portfolios $\Pi_\stoch, \Pi_\det$ and of the portfolio $\Pi_{\none}$.}
\label{fig:CTDhedge-constant-SD}
\end{figure}
\begin{figure}[]
\centering
\includegraphics[width=.49\textwidth, height=0.40833\textwidth]{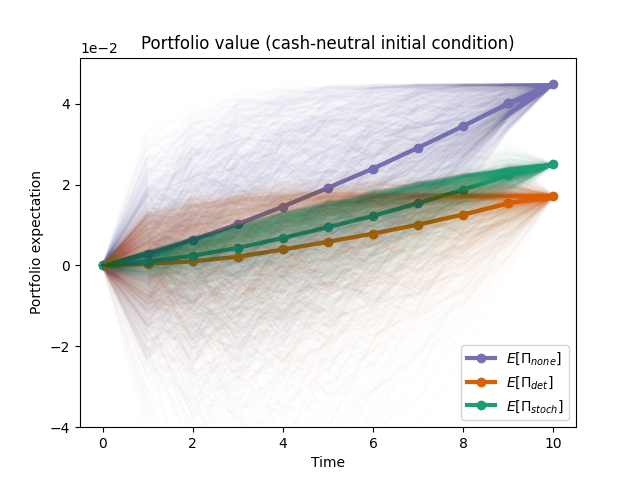}
\includegraphics[width=.49\textwidth, height=0.40833\textwidth]{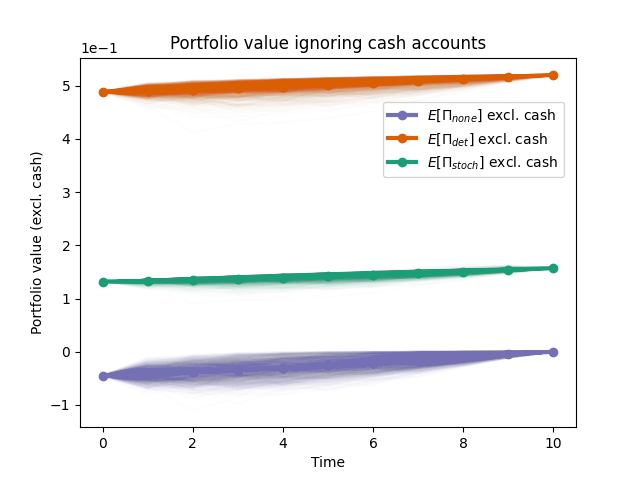}
\caption{Valuation of the CTD-hedging portfolios $\Pi_\stoch(t), \Pi_\det(t)$ and portfolio $\Pi_\none(t)$. Sample paths and expected values are provided. In the graph on the right, the portfolios are stripped of their respective cash accounts, eliminating the initial zero condition.}
\label{fig:CTDhedge-constant-E}
\end{figure}
In the first experiment, we consider a scenario of non-crossing collateral spread forecasts. That is, the deterministic collateral spreads $\hat q_i(t)$, $i\in\{0,1,2\}$, which are also the expectations of the stochastic dynamics, are chosen such that they do not cross and $\hat q_1(t)$ remains maximal for all $t$ until maturity $T=10$. A graph of the spread configuration is presented on the left in \Cref{fig:CTDhedge-constant-SD}.
\par
As there are no crossing times after $t_0$, the deterministic strategy portfolio is given by
\begin{equation}\label{eq:det-ex1}
\Pi_\det(t) = P^c(t, T) - Q_1(t, T) + \Bigl(-P^c(t_0, T) + Q_1(t_0, T)\Bigr).
\end{equation}
This is equal to the basic portfolio $\Pi_1$ and due to the clear hierarchy of collateral spreads, we abstain from simulating the basic portfolio $\Pi_2$. We do consider the basic, domestic portfolio $\Pi_0(t) = \Pi_\none(t)$, as it describes a portfolio unconcerned with the collateral choice option. As here $r_0(t) \equiv 0$, it is equal to
\begin{equation}\label{eq:Pinoneexperiment}
\Pi_\none(t) = P^c(t, T) - 1 + (-P^c(t_0, T) + 1 \Bigr).
\end{equation}
\par
The stochastic strategy portfolio is obtained with the variance minimization \eqref{eq:minimizevar}.
We note that in the constant zero domestic interest rate scenario, the variance to minimize is given by
\begin{align}\label{eq:alpha0var}
\Var\bigl[\widetilde \pi(\alpha)\bigr] &= \Var\Bigl[\widetilde P^c(t_0, T) + \alpha_0 \widetilde Q_0(t_0, T) + \alpha_1 \widetilde Q_1(t_0, T) + \alpha_2 \widetilde Q_2(t_0, T)\Bigr] \nonumber \\
&= \Var\Bigl[\widetilde P^c(t_0, T) + \alpha_1 \widetilde Q_1(t_0, T) + \alpha_2 \widetilde Q_2(t_0, T)\Bigr],
\end{align}
with the random variables defined in \Cref{defi:performancevar}. The latter equality follows from the particular scenario with $r_0(t) \equiv 0$.
This implies that the variance minimization approach leaves weight $\alpha_0$ unspecified because of the lack of stochasticity in $Q_0(t, T) = P(t, T) = 1$ of this particular scenario. However, the lack of interest accrual also implies that the weight $\alpha_0$ does not enter the portfolio with initial value zero in a meaningful way, as it cancels out with the bank account, shown in \eqref{eq:stoch-ex1} below. The weights obtained from the variance minimization procedure are $\alpha_1 \approx -0.477$ and $\alpha_2 \approx -0.361$ and 
%
the stochastic strategy portfolio is thus equal to
\begin{align}\label{eq:stoch-ex1}
\Pi_\stoch(t) 
&= P^c(t, T)  + \alpha_0 - 0.477\, Q_1(t, T) - 0.361\, Q_2(t, T) \nonumber \\
&+ \Bigl(-P^c(t_0, T) - \alpha_0 + 0.477\, Q_1(t_0, T) + 0.361\, Q_2(t_0, T)\Bigr) \nonumber \\
&= P^c(t, T) - 0.477\, Q_1(t, T) - 0.361\, Q_2(t, T) \nonumber \\
&+ \Bigl(-P^c(t_0, T) + 0.477\, Q_1(t_0, T) + 0.361\, Q_2(t_0, T)\Bigr).
\end{align}
\par
The portfolios $\Pi_\det(t)$, $\Pi_\stoch(t)$ and $\Pi_\none(t)$ are evaluated in a Monte Carlo market simulation based on Hull--White dynamics. Market prices of the involved instruments,
$\E[Q_1(t_0, T)]$, $\E[Q_2(t_0, T)]$, and $\E[P^c(t_0, T)]$, can be obtained from the Monte Carlo simulation or, similarly, from analytical Hull--White ZCB formulas and the common factor CTD approach, respectively.
\par
In the right graph of \Cref{fig:CTDhedge-constant-SD}, the standard deviations of the portfolio simulations, defined as $\mathrm{SD}[\Pi(t)] := \sqrt{\Var[\Pi(t)]}$,  are evaluated up to maturity. As all considered portfolios have deterministic payoffs at maturity, all the standard deviations taper back to zero at time $T$. It is evident that the stochastic strategy portfolio displays by far the least standard deviation. This comes as expected since the portfolio was constructed from principles of variance minimization. The deterministic strategy hedging portfolio still displays less standard deviation than the CTD-indifferent portfolio $\Pi_\none$.
\par
In \Cref{fig:CTDhedge-constant-E} on the left, valuations of the hedging portfolios are given for portfolio sample paths and expected portfolio values. We clearly observe the initial condition $\Pi(0) =0$ for all portfolios. 
A hypothetical, perfect hedging portfolio could be obtained from the back-to-back hedge with a zero-coupon bond that has the same collateral choice option{. 
This perfect portfolio vanishes throughout its lifetime.} Accordingly, we can associate the valuation of a portfolio at time $t$ with the cost of the imperfect hedging at that time. We are particularly interested in the valuations at maturity $T$, which are deterministic. Notably, the CTD-indifferent hedging portfolio $\Pi_\none(T) \approx 0.045$ carries the highest penalty, whereas the stochastic strategy portfolio $\Pi_\stoch(T) \approx 0.025$ performs slightly worse than the deterministic strategy $\Pi_\det(T) \approx 0.017$. 
In the graph on the right of \Cref{fig:CTDhedge-constant-E}, the portfolios are given without the requirement of initial price zero, i.e. $\bar \Pi_j(t) = \Pi_j(t) - \mathcal{C}_j(t)$ for $j\in\{\stoch,\det,\none\}$ and $\mathcal{C}_j$ the corresponding cash accounts. The initial valuations of these portfolios, $\bar \Pi_j(0)$, indicate the amount of cash which needs to be borrowed to enter the respective hedging portfolios. Since the uncollateralized bond $P(t_0, T)\geq P^c(t_0, T)$ is never priced lower than the collateral-choice bond, entering the strategy $\Pi_\none$ yields non-negative cash. 
In contrast, entering the deterministic hedging strategy requires funding, $\bar \Pi_\det(0) = 0.49$.
The stochastic hedging strategy takes a special role, as here the so far unspecified weight $\alpha_0$ comes into play. We can indeed choose the weight $\alpha_0$ to obtain an arbitrary initial price,
\begin{align}
\bar \Pi_\stoch(t_0) &= P^c(t_0, T)  + \alpha_0 - 0.477\, Q_1(t_0, T) - 0.361\, Q_2(t_0, T) \nonumber \\
&= +\alpha_0 + 0.132.
\end{align}
An initial choice of $\alpha_0 = -0.132$ yields a cash-neutral portfolio at inception, in which case the valuation paths are again those of the portfolio shown in the graph on the left.
%
 %
\par
\begin{figure}[]
\centering
\includegraphics[width=.49\textwidth, height=0.40833\textwidth]{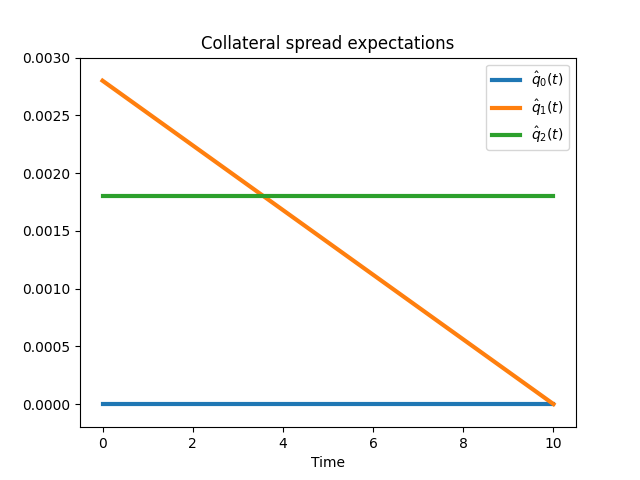}
\includegraphics[width=.49\textwidth, height=0.40833\textwidth]{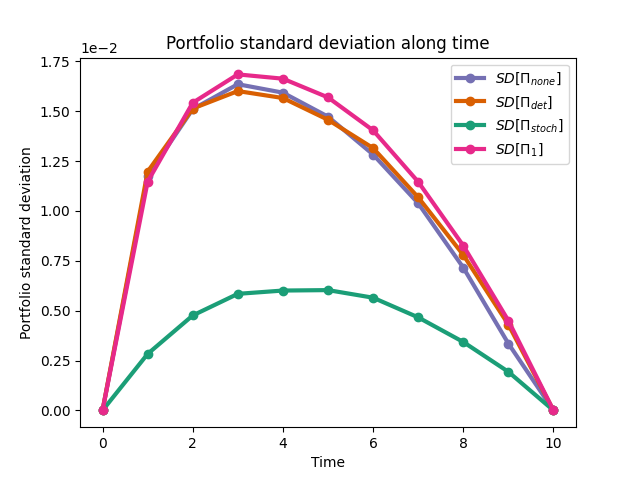}
\caption{In the second CTD-hedging experiment, the deterministic collateral spreads (respectively expectations of collateral spreads) cross.}
\label{fig:CTDhedge-crossing-SD}
\end{figure}
\begin{figure}[]
\centering
\includegraphics[width=.49\textwidth, height=0.40833\textwidth]{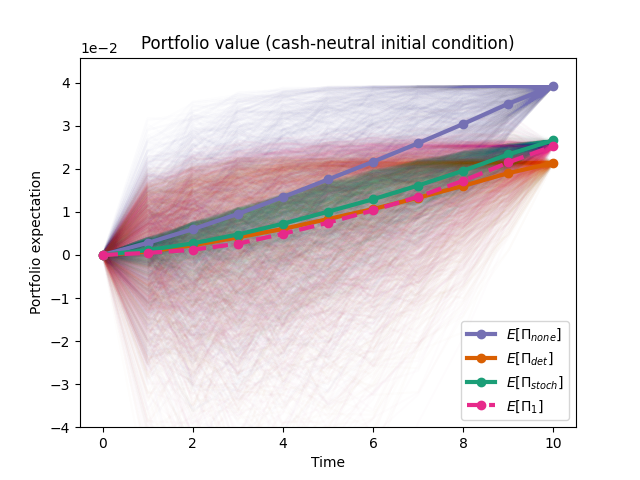}
\includegraphics[width=.49\textwidth, height=0.40833\textwidth]{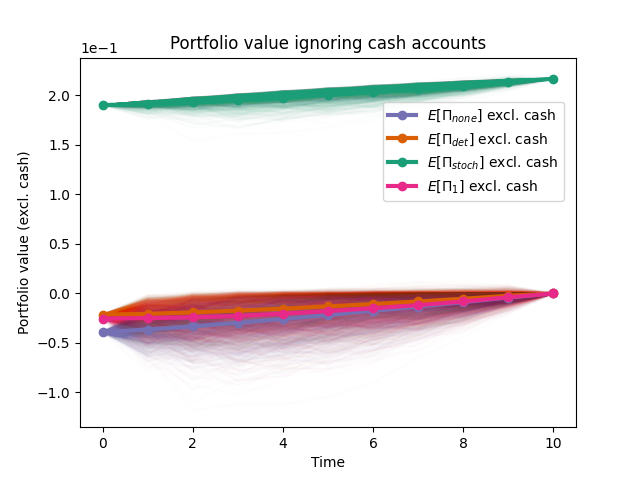}
\caption{Portfolio valuations in the second CTD-hedging experiment, where the deterministic collateral spreads cross.}
\label{fig:CTDhedge-crossing-E}
\end{figure}
In the second hedging experiment, we consider a scenario where the projections for the collateral spreads cross. Until the crossing time $T^C_1 = 3.6$, the first collateral spread $\hat q_1(t)$ is maximal, afterwards the second collateral spread $\hat q_2(t)$ is. This behaviour is depicted in the left graph of \Cref{fig:CTDhedge-crossing-SD}. The resulting deterministic strategy portfolio is given by 
\begin{align}
\Pi_\det(t) &= P^c(t, T) - Q_1(t, T) + F_1(t, T^C_1, T) - F_2(t, T^C_1, T) \nonumber \\
&+ \Bigl(-P^c(t_0, T) + Q_1(t_0, T) - F_1(t_0, T^C_1, T) + F_2(t_0, T^C_1, T)\Bigr).
\end{align}
Since the domestic interest rate is constant zero, we have $P(t, S) = 1$ for all $S\in[t_0, T]$ and therefore the forward contracts simplify to $F_i(t, S, T) = Q_i(t, T)$. Thus,  in this setting, the deterministic strategy is equal to the basic strategy $\Pi_2(t)$ with hedging instrument $Q_2(t, T)$.
We furthermore add the basic strategy with hedging portfolio $\Pi_1$, given by
\begin{equation}
\Pi_1(t) = P^c(t, T) - Q_1(t, T) + \Bigl(-P^c(t_0, T) + Q_1(t_0, T) \Bigr),
\end{equation}
and the hedging strategy which ignores the CTD-discount factor, $\Pi_\none(t)$, given in \eqref{eq:Pinoneexperiment}.
\par
From portfolio variance minimization, we obtain the hedging weights of the stochastic strategy, $\alpha_1\approx- 0.343$, $ \alpha_2 \approx-0.436$ and the corresponding stochastic strategy portfolio
\begin{align}
\Pi_\stoch(t) &= P^c(t, T) - 0.343\, Q_1(t, T) -0.436\, Q_2(t, T) \nonumber \\
&+ \Bigl(-P^c(t_0, T) + 0.343\, Q_1(t_0, T) +0.436\, Q_2(t_0, T)\Bigr).
\end{align}
We refer to the discussion around \eqref{eq:stoch-ex1} regarding the unspecified parameter $\alpha_0$.
Comparing the weights $(\alpha_1, \alpha_2)$ to the weights obtained in the first experiment, we notice increased importance attributed to the second foreign zero-coupon bond $Q_2(t_0, T)$ in the portfolio. This aligns with an increased probability of $q_2(t)$ being the maximal spread, which can easily be analytically confirmed with the updated collateral spread expectations under unchanged volatilities $\xi_1, \xi_2$.
\par
The standard deviations of the portfolios observed in this experiment, displayed in \Cref{fig:CTDhedge-crossing-SD}, resemble those of \Cref{fig:CTDhedge-constant-SD} in magnitude, with the stochastic portfolio clearly displaying the least variance. However, all deterministic strategies, basic and maximal-spread-following, exhibit more deviation than before, now at similar levels to the strategy of portfolio $\Pi_\none$.
\par
In \Cref{fig:CTDhedge-crossing-E}, we have the portfolio valuations for this scenario, again with the bank account ensuring an initial price of zero and without. Initial-zero portfolios behave similar to the observation of the previous experiment, with the portfolio $\Pi_\none(T) \approx 0.039$ exhibiting the highest value at time $T$, which we understand as the penalty for the imperfect hedge. As before, the deterministic strategy ($\Pi_\det(T) \approx 0.021$) is slightly cheaper than the stochastic strategy ($\Pi_\stoch(T) \approx 0.027$) and the basic strategy ($\Pi_1(T) \approx 0.025$).
\par
The right graph of \Cref{fig:CTDhedge-crossing-E}, which shows the portfolio valuations when the cash account is removed, presents a notably different picture from before. Except for the stochastic strategy, all hedging portfolios are cash-positive at inception. As remarked before, the stochastic portfolio is subject to arbitrary, parallel shifts about the vertical axis by the choice of unspecified parameter $\alpha_0$. The cash-neutral choice, in this case, is given by $\alpha_0 \approx -0.19$.
\section{Conclusions}\label{sec:conclu}
We have demonstrated structural differences between the deterministic and stochastic models for the collateral choice option, and we have shown that assets with the collateral choice option require tailored hedging strategies when no hedging instruments equipped with the same option are available. To this end, we have proposed hedging strategies based on the domestic and foreign zero-coupon bonds associated with the available collateral currencies. A variance-minimizing strategy was proposed based on the stochastic collateral spread model, and we have shown how its hedging weights can be obtained semi-analytically with an application of the common factor approach to CTD discount factor pricing. We have further proposed hedging strategies based on the deterministic collateral spread model and its implied collateral spread crossing times.
Several numerical experiments have been performed in a setting where domestic risk factors are removed, resulting in a scenario where only risks from the alternative collateral currencies enter. The numerical results affirm that the strategy derived from the stochastic model strongly reduces the portfolio variance compared to all other strategies considered. In any case, hedging with a currency associated with a positive collateral spread improved hedging performance over the indiscriminate choice of the domestic currency. 
Future research should consider the effect of the collateral choice option on non-linear products in a broadened scenario where all risk factors are present in the market. 
Furthermore, it remains an important question how the strong assumption of free collateral substitution can be loosened without making the collateral choice option price path dependent and how the collateral spreads can be calibrated in an efficient way.
%
%
%
\section*{Acknowledgments}
This research is part of the ABC--EU--XVA project and has received funding from the European Union's Horizon 2020 research and innovation programme under the Marie Sk\l{}odowska--Curie grant agreement No.\ 813261.

\bibliography{cf-Lit}



\appendix
\section{Proof of \Cref{lem:thetalemma}}\label{appx:thetalemma}
\begin{proof}Throughout this proof, expectations are taken under either the domestic $\Q_0$ measure or the domestic forward $\Q_0^T$ measure, according to the definition of the collateral spreads in \eqref{eq:HWspread}.
\par
1. Let $\theta_i(t)$ be defined as in \eqref{eq:newexacttheta}. Then it holds for every $t\in[t_0, T]$:
	\begin{align}
	\E[q_i(t)] &= q_i(t_0) \e^{-\kappa_i (t-t_0)} + \kappa_i \int_{t_0}^t \theta_i(z) \e^{-\kappa_i (t-z)} \d z \nonumber \\
	&= \hat q_i(t_0)\e^{-\kappa_i (t-t_0)} + \kappa_i \int_{t_0}^t \bigl(\hat q_i(z) + \frac\partial{\partial z} \frac1{\kappa_i} \hat q_i(z)\bigr) \e^{-\kappa_i (t-z)} \d z \nonumber \\
	&= \hat q_i(t_0)\e^{-\kappa_i (t-t_0)} + \e^{-\kappa_i t} \Bigl(\int_{t_0}^t \hat q_i(z)\kappa_i\e^{\kappa_i z} \d z + \int_{t_0}^t \frac\partial{\partial z} \hat q_i(z) \e^{\kappa_i z} \d z\Bigr).
	\end{align}%
The result follows from the integration by parts formula
\begin{equation}
\int_{t_0}^t \hat q_i(z)\kappa_i\e^{\kappa_i z} \d z + \int_{t_0}^t \frac\partial{\partial z} \hat q_i(z) \e^{\kappa_i z} \d z 
= \Bigl[ \hat q_i(z) \e^{\kappa_i z}\Bigr]^t_{t_0}
= \hat q_i(t)\e^{\kappa_i t} - \hat q_i(t_0)\e^{\kappa_i t_0}.
\end{equation}
\par
2. Let $\mathcal{T}=\{t_0, t_1, \dots, t_R = T\}$ be the time discretization of $[t_0, T]$ and let $\theta_i(t)$ be defined piecewise constant on the intervals $(t_{k-1}, t_k]$ for every $k\geq 1$ as described in \eqref{eq:newtheta}. For every $t_k$, $k\geq 1$ it holds that
\begin{align}
	\E[q_i(t_k)] &= q_i(t_0)\e^{-\kappa_i (t_k-t_0)} + \kappa_i \int_{t_0}^{t_k} \theta_i(z) \e^{-\kappa_i (t_k-z)} \d z \nonumber \\
%
%
	&= \hat q_i(t_0)\e^{-\kappa_i (t_k-t_0)} + \sum\limits_{j=1}^k \kappa_i \e^{-\kappa_i t_k} \frac{\hat q_i(t_j) \e^{\kappa_i t_j} - \hat q_i(t_{j-1}) e^{\kappa_i t_{j-1}}}{\e^{\kappa t_j} - \e^{\kappa t_{j-1}}}
					\int_{t_{j-1}}^{t_j} \e^{\kappa_i z} \d z \nonumber \\
	&= \hat q_i(t_0)\e^{-\kappa_i (t_k-t_0)} + \e^{-\kappa_i t_k} \sum\limits_{j=1}^k \bigl( \hat q_i(t_j) \e^{\kappa_i t_j} - \hat q_i(t_{j-1}) e^{\kappa_i t_{j-1}} \bigr) \nonumber \\
	&= \hat q_i(t_k),
		\end{align}%
		where the final step is based on a telescopic sum.
\end{proof}
%
\section{Computation of the portfolio variance minimization}
\label{appx:varianceminimization}
In the following, we give the mathematical details of the portfolio variance minimization approach. By the end of this section, we obtain a semi-analytical approach to the hedging weights obtained from the variance minimization in \eqref{eq:minimizevar}.
%
%
The first result is a transformation of the variance $\Var[\pi(\alpha)]$ into the object of interest $f(\alpha)$ given in \eqref{eq:f-varmin}.
\begin{lem}
Minimizing the variance of the random variable $\pi(\alpha)$ given in \Cref{defi:performancevar}, is equivalent to minimizing
\begin{align}\label{eq:f-varminAPPX}
f(\alpha) = &\sum_{i=0}^N \alpha_i^2 \Var[\widetilde Q_i(t_0, T)] +  \sum_{\overset{i,j=0}{j\neq i}}^N \alpha_i \alpha_j \Cov[\widetilde Q_i(t_0, T), \widetilde Q_j(t_0, T)] \nonumber \\
&+ 2 \sum_{i=0}^N \alpha_i \Cov[\widetilde P^c(t_0, T), \widetilde Q_i(t_0, T)].
\end{align}
\end{lem}
In the following, we do not write the repeated arguments $(t_0, T)$ of the random variables $\widetilde Q_i(t_0, T)$ and $\widetilde P^c(t_0, T)$.
\begin{proof}
The proof is obtained by expansion of the variance of $\pi(\alpha)$,
\begin{align}
\Var[\pi(\alpha)] &= \Var[\widetilde P^c +\sum_{i=0}^N \alpha_i \widetilde Q_i]  \nonumber \\
&= \Var[\widetilde P^c] + \sum_{{i,j=0}}^N \alpha_i \alpha_j \Cov[\widetilde Q_i, \widetilde Q_j] +  2 \sum_{i=0}^N \alpha_i \Cov[\widetilde P^c,\widetilde Q_i]. 
\end{align}
The result follows by the removal of the first term, $\Var[\widetilde P^c]$, which does not contain any weights $\alpha_i$, $i\geq 0$.
\end{proof}
Towards a semi-analytical expression of the function $f(\alpha)$, we begin by analysing the covariance term $\Cov[\widetilde Q_i, \widetilde Q_j]$, which can, for all $i, j \geq 0$, be expressed  by 
\begin{align}\label{eq:covQiQj}
\Cov[\widetilde Q_i, \widetilde Q_j] &= \E^{\Q_0}\bigl[ \bigl(\e^{-\int_{t_0}^T q_i(s)\d s}\e^{-\int_{t_0}^T r_0(s)\d s} \bigr)\bigl(\e^{-\int_{t_0}^T q_j(s)\d s}\e^{-\int_{t_0}^T r_0(s)\d s} \bigr)\bigr] \nonumber \\
&\quad - \E^{\Q_0}\bigl[\e^{-\int_{t_0}^T q_i(s)\d s}\e^{-\int_{t_0}^T r_0(s)\d s}\bigr]\E^{\Q_0}\bigl[\e^{-\int_{t_0}^T q_j(s)\d s}\e^{-\int_{t_0}^T r_0(s)\d s}\bigr] \nonumber \\
&= \E^{\Q_0}\bigl[\e^{-\int_{t_0}^T q_i(s)\d s}\e^{-\int_{t_0}^T q_j(s)\d s}\bigr]\E^{\Q_0}\bigl[ \bigl(\e^{-\int_{t_0}^T r_0(s)\d s}\bigr)^2\bigr] \nonumber \\
&\quad - \E^{Q_0}\bigl[\e^{-\int_{t_0}^T q_i(s)\d s}\bigr]\E^{Q_0}\bigl[\e^{-\int_{t_0}^T q_j(s)\d s}\bigr]\E^{Q_0}\bigl[\e^{-\int_{t_0}^T r_0(s)\d s}\bigr]^2.
\end{align}
Above, we used the assumption of independence under $\Q_0$ between the collateral spreads and domestic interest rate which was made in \Cref{ref:ctdhedging}.
We note that \eqref{eq:covQiQj} holds even for $i=0$ or $j=0$, since $q_0(t) = 0$ and thus $\E[\e^{-\int_{t_0}^T q_0(s)\d s}] = 1$.
\par
The collateral spreads $q_i(t)$, $i\geq 1$ are defined as 1-factor Hull--White processes under the particular choice of measure $\Q_0$ in \eqref{eq:HWspread}, and $r_0(t)$ is defined as an analogous 1-factor Hull--White process under the measure $\Q_0$ in \eqref{eq:r0dynamics}. Therefore, the first moments $\E^{\Q_0}[\e^{-\int_{t_0}^T q_i(s)\d s}]$ and $\E^{\Q_0}[\e^{-\int_{t_0}^T r_0(s)\d s}]$ in \eqref{eq:covQiQj} can be directly obtained from the well-known pricing formula for a zero-coupon bond under the Hull--White model, see, for example, \cite{GrzelakOosterlee}.
\par
Let us now concentrate on the second-order terms. In the following lemma, we consider the mixed term $\E^{\Q_0}\bigl[\e^{-\int_{t_0}^T q_i(s)\d s}\e^{-\int_{t_0}^T q_j(s)\d s}\bigr]$; the second moment of $\widetilde Q_0$ can be derived similarly.
\begin{lem}
The mixed second-order term $\E^{\Q_0}\bigl[\e^{-\int_{t_0}^T q_i(s)\d s}\e^{-\int_{t_0}^T q_j(s)\d s}\bigr]$ is, for every $i, j \geq 1$, given by
\begin{equation}\label{eq:lognormal}
\E^{\Q_0}\bigl[\e^{-\int_{t_0}^T q_i(s)\d s}\e^{-\int_{t_0}^T q_j(s)\d s}\bigr] = \exp\Bigl( \mu_{i,j} + \frac12 v_{i,j}\Bigr),
\end{equation}
where
\begin{equation}\label{eq:lognormalexpec}
\mu_{i,j} := - \E^{\Q_0}\bigl[\int_{t_0}^T \bigl(q_i(s) + q_j(s)\bigr) \d s\bigr] = - \int_{t_0}^T \bigl(\hat q_i(s) + \hat q_j(s)\bigr) \d s,
\end{equation}
and
\begin{equation}\label{eq:lognormalvar}
v_{i,j} := \Var\bigl[\int_{t_0}^T q_i(s)\d s\bigr] + \Var\bigl[\int_{t_0}^T q_j(s)\d s\bigr] + 2\, \Cov\bigl[\int_{t_0}^T q_i(s)\d s, \int_{t_0}^T q_j(s)\d s\bigr].
\end{equation}
Furthermore, it holds for every $i,j\geq 1$,
\begin{align}\label{eq:spreadintegcov}
\Cov\bigl[\int_{t_0}^T q_i(s) \d s, \int_{t_0}^T q_j(s) \d s\bigr] = &\frac{\xi_i \xi_j \rho_{i,j}}{\kappa_i + \kappa_j} \Biggl( \frac{\e^{-\kappa_j \tau}}{\kappa_j^2} - \frac{1}{\kappa_j^2} + \frac{\tau}{\kappa_j} +\frac{1}{\kappa_i^2} - \frac{\e^{-\kappa_i \tau}}{\kappa_i^2} + \frac{\tau}{\kappa_i} \nonumber \\
&\quad - \frac{\e^{-(\kappa_i + \kappa_j)\tau}}{\kappa_i\kappa_j} + \frac{\e^{-\kappa_i \tau}}{\kappa_i \kappa_j} 
+ \frac{e^{-\kappa_j \tau}}{\kappa_i\kappa_j} - \frac{1}{\kappa_i \kappa_j}\Biggr), 
\end{align}
where $\tau := T - t_0$, and all variances, covariances are taken under the $\Q_0$ measure.
\end{lem}
\begin{proof}
Since every $q_i(s)$ is a normal random variable, so is $q_i(s) + q_j(s)$ and also $\int_{t_0}^T q_i(s) + q_j(s)\d s$. Consequently, $\E^{\Q_0}[\exp(-\int_{t_0}^T q_i(s) + q_j(s)\d s)]$ is the expectation of a log-normal random variable,  thus \eqref{eq:lognormal} holds and \eqref{eq:lognormalexpec}, \eqref{eq:lognormalvar} immediately follow. Finally, \eqref{eq:spreadintegcov} is obtained from the observation 
\begin{equation}
\Cov\bigl[\int_{t_0}^T q_i(s) \d s, \int_{t_0}^T q_j(s) \d s\bigr] = \int_{t_0}^T \int_{t_0}^T \Cov[q_i(u), q_j(v)] \d u \d v,
\end{equation}
and the explicit solution of $\Cov[q_i(u), q_j(v)]$, given by
\begin{align}\label{eq:covqiqj}
\Cov\bigl[q_i(u), q_j(v)\bigr] &= \xi_i \xi_j \e^{-(\kappa_i u + \kappa_j v)} \Cov \bigl[ \int_{t_0}^u \e^{\kappa_i z} \d W_i(z), \int_{t_0}^v \e^{\kappa_j y} \d W_j(y)\bigr] \nonumber \\
&= \frac{\xi_i\xi_j\rho_{i,j}}{\kappa_i + \kappa_j} \e^{-(\kappa_i u + \kappa_j v)} \Bigl( \e^{(u \wedge v)(\kappa_i + \kappa_j)} - \e^{t_0 (\kappa_i + \kappa_j)} \Bigr).
\end{align}
The latter result is obtained by using that  $W_j \overset{d}= \rho_{i,j}W_i + \sqrt{1 - \rho_{i,j}^2} W^\bot$, with a Brownian motion $W^\bot$ independent of $W_i$, which enables an application of It\^o's isometry. 
\end{proof}
This concludes the analytical expressions in the variance minimization approach. It remains to find an expression for the term $\Cov[\widetilde P^c,\widetilde Q_i]$, $i\geq 0$.
\par
We begin with the case $i=0$. Following the same factorization argument as in \eqref{eq:covQiQj}, we obtain
\begin{equation}
\Cov[\widetilde P^c, \widetilde Q_0] = \E^{\Q_0}\bigl[\e^{-\int_{t_0}^T M(s) \d s}\bigr]\Bigl( \E^{\Q_0}\bigl[\widetilde Q_0^2\bigr] -  \E^{\Q_0}\bigl[\widetilde Q_0\bigr]^2 \Bigr),
\end{equation}
where it was defined that $M(s) = \max(0, q_1(s), \dots, q_N(s))$. We showed previously how the moments of $\widetilde Q_0$ can be obtained. The remaining term is precisely the CTD discount factor $\CTD(t_0, T)$ and we can use the common factor approximation, given in \eqref{eq:CFCTD}, to obtain a semi-analytical approximation.
\par
For the case that $i\geq 1$, it holds by the same line of reasoning that
\begin{align}
\Cov\bigl[\widetilde P^c, \widetilde Q_i\bigr] &= \E^{\Q_0}\bigl[\e^{-\int_{t_0}^T M(s) \d s} \widetilde Q_i \widetilde Q_0 \bigr] - \E^{\Q_0}\bigl[\e^{-\int_{t_0}^T M(s) \d s}\bigr] \E^{\Q_0}\bigl[\widetilde Q_i \bigr] \E^{\Q_0}\bigl[ \widetilde Q_0 \bigr].
%
\end{align}
Most of these terms has been previously treated, except for $\E^{\Q_0}[\e^{-\int_{t_0}^T M(s) \d s} \widetilde Q_i \widetilde Q_0]$. The missing term is
\begin{align}
&\E^{\Q_0}\bigl[\e^{-\int_{t_0}^T M(s) \d s} \widetilde Q_i \widetilde Q_0\bigr] \nonumber \\
&= \E^{\Q_0}\Bigl[ \e^{-\int_{t_0}^T \max\bigl(q_i(s), q_1(s) + q_i(s), \dots, q_N(s) + q_i(s)\bigr)} \e^{-\int_{t_0}^T 2r_0(s)\d s}\Bigr] \nonumber \\
&= \E^{\Q_0}\Bigl[ \e^{-\int_{t_0}^T \max\bigl(q_i(s), q_1(s) + q_i(s), \dots, q_N(s) + q_i(s)\bigr)}\Bigr]\E^{\Q_0}\Bigl[ \e^{-\int_{t_0}^T 2r_0(s)\d s}\Bigr], 
\end{align}
which we recognize as closely related to the CTD discount factor. The procedure of the common factor approximation, detailed in \cite{WolfCF}, can be adapted for this special case. Herefore, it is only necessary to construct the joint marginal distributions
\begin{equation}
\bigl(q_i(s), q_1(s) + q_i(s), \dots, q_N(s) + q_i(s) \bigr) \sim \mathcal{N}\Bigl( \mu(s; i), \Sigma(s; i) \Bigr),
\end{equation}
where
\begin{equation}
\mu(s; i) := \bigl(\hat q_i(s), \hat q_1(s) + \hat q_i(s), \dots, \hat q_N(s) + \hat q_i(s)\bigr),
\end{equation}
and the covariance matrix $\Sigma(s;i)$ can be directly obtained from \eqref{eq:covqiqj}.
\par
We have thus shown how the variance minimization term can be expressed with a mixture of analytical results and semi-analytical approximations. Terms involving only the random variables $\widetilde Q_i$, associated with domestic and foreign zero-coupon bonds, are solved analytically and terms involving the random variable $\widetilde P^c$, associated with the asset equipped with a collateral choice option, are approximated semi-analytically.

\end{document}